\documentclass[10pt,aps,pra,twocolumn,superscriptaddress,showpacs,showkeys,amsmath,amssymb]{revtex4-1}
\usepackage{amsmath}
\usepackage{accents}
\usepackage{bm}
\usepackage{quantikz}
\usepackage{amsthm}
\usepackage{dsfont}
\usepackage[english]{babel}
\usepackage{graphicx}
\usepackage{adjustbox}
\usepackage{color}
\usepackage{colortbl}
\usepackage{mathtools}
\usepackage{commath}
\usepackage[justification=justified]{caption}

\usepackage{subcaption}
\usepackage{floatrow}
\usepackage[caption=false]{subfig}
\usepackage[toc,page,title]{appendix}
\usepackage{multirow}

\newtheorem{theorem}{Theorem}[]
\newtheorem{lemma}[theorem]{Lemma}

\begin{document}

\title{Title}

\author{Vladyslav Bohun}
\affiliation{Haiqu, Inc., 95 Third Street, San Francisco, CA 94103, USA}
\author{Andrij Kuzmak}
\affiliation{Haiqu, Inc., 95 Third Street, San Francisco, CA 94103, USA}
\affiliation{Department for Theoretical Physics, Ivan Franko National University of Lviv, Lviv, Ukraine}
\author{Maciej Koch-Janusz}
\affiliation{Haiqu, Inc., 95 Third Street, San Francisco, CA 94103, USA}
\affiliation{Department of Physics, University of Z\"urich, 8057 Z\"urich, Switzerland}

\title{Quantum Hamiltonian simulation of linearised Euler equations in complex geometries}
\begin{abstract}
Quantum computing promises exponential improvements in solving large systems of partial differential equations (PDE), which forms a bottleneck in high-resolution simulations of, among others, computational fluid dynamics (CFD) in aerospace applications and weather forecasting. One approach is via mapping classical PDE problems to a quantum Hamiltonian evolution, for which recently an explicit quantum circuit construction has been shown in simple cases, allowing proof-of-concept execution on quantum processors. Here we extended this method to more complex and practically relevant cases. We first demonstrate how arbitrary complex-shaped geometrical obstacles with Dirichlet, Neumann or mixed boundary conditions can be introduced in the quantum representations of elementary difference operators used to implement the PDE, either directly or using Linear Combination of Unitaries (LCU). We provide their explicit and efficient circuit constructions, and analyze the Trotter errors and asymptotic gate complexities, which in the Dirichlet case do not grow compared to the free space equation. Using these methods we then derive quantum circuits for the linearized Euler equations in the presence of a background fluid flow and obstacles, both in the conservative and non-conservative regimes. We illustrate our results by simulating the obtained quantum circuits for different boundary conditions and geometries, comparing their error to a classical finite difference scheme.
\end{abstract}

\maketitle

\section{Introduction \label{sec1}}
Computational Fluid Dynamics (CFD) methods are key to understanding and modeling natural phenomena of fundamental importance in both science and engineering, from weather and climate forecasting, through biological flows, to aircraft and automobile design. The computational complexity of solving the associated partial differential equation (PDE) systems on required resolutions and sufficiently long time-scales poses, however, a formidable challenge to even state-of-art High-Performance Computing systems.

The progress in quantum computing hardware \cite{arute2019quantum,bluvstein2024logical,google2025quantum} and foundational algorithms \cite{childs2012hamiltonian,Cao2013,PhysRevLett.118.010501,PhysRevLett.103.150502,PRXQuantum.2.040203,gilyen2019quantum,motlagh2024generalized} has inspired a flurry of theoretical research into solving PDEs on a quantum computer \cite{zamora2024efficient,Engel_2019,berry2014,Budinski_2022,PhysRevA.93.032324,berry2017,PhysRevA.100.062315,linden2020,childsliu2020qs,Childsliuostrander2021,golse2022,jinliuyue2023} over the past decade, motivated by the exponential speedups achievable. 
The proposed quantum approaches span a large space, and differ in their domain applicability (\emph{e.g.}~steady state or evolution problems), theoretical guarantees, and suitability for execution on devices prior to achieving full fault-tolerance (\emph{e.g.} HHL \cite{Cao2013,PhysRevLett.103.150502} \emph{vs.}~variational approaches \cite{leong2023variational_pde,Lubasch_2020_variational_pde2,Ali_2023_poisson_pde,PhysRevA.104.052409_variational_poisson,Cerezo2021,TILLY20221,Sato_2023_eigensolver_pde}). We focus on time-dependent partial differential equations.

One of the most promising approaches to solving such time-dependent problems is to recast them as a quantum time-evolution described by a Schr\"odinger equation, with an auxiliary Hamiltonian defined by the PDE \cite{toyoizumi2023hamiltoniansimulationusingquantum,PhysRevX.13.041041,Jin2022,jin2023quantumsimulationdiscretelinear,toyota2024pde,satohiroukinaoki2024}. The evolution is then implemented on a quantum computer as a quantum circuit, using \emph{e.g.}~Lie-Trotter-Suzuki method. Whether the problem is easily expressible as Hamiltonian operator is heavily dependent on the PDE; in general case it may require introducing auxiliary qubits and methods such as block-encoding and linear combinations of unitaries \cite{jin2025schrodingerizationbasedquantumalgorithms,jin2023quantumsimulationquantumdynamics,PhysRevLett.133.230602,leamer2024quantumdynamicalemulationimaginary,guseynov2024explicitgateconstructionblockencoding,Novikau_2025,An_2023}. While such constructions often require prohibitively large resources, recently a number of proposals have appeared for implementations in specific cases, in particular for the wave-equation \cite{toyota2024pde,PhysRevA.99.012323,Suau_2021,Wright_2024}, some of which are more suited to near term devices.

Here we tackle a more general problem. As most practically relevant CFD simulations require complicated boundary conditions describing \emph{e.g.}~objects such as wings or constrictions in the flow, we first show how in the mapping to a Hamiltonian evolution arbitrary obstacles can be efficiently introduced in the quantum representations of the PDE. We build explicit quantum circuits allowing their implementation using the Trotter approach for different types of boundary conditions. In the Dirichlet case we prove that our construction neither increases the Trotter error of the simulation nor the asymptotic gate complexity of the simulation circuits, while for more general boundary conditions we give a Linear Combination of Unitaries (LCU) implementation, whose overheads are still polynomial. Second, using the above results, we develop explicit quantum circuits for the important problem of modeling sound propagation and scattering of obstacles in the presence  of a constant background fluid flow, which is described by the linearized Euler equations (LEE). We verify our constructions by simulating the obtained quantum circuits for different boundary conditions, and comparing the results to those of classical finite difference methods (FDM). 

The paper is organized as follows: in Sec.~\ref{prelim} we review how finite difference operators can be represented as quantum circuits, and how, consequently, PDEs can be represented as a Schr\"odinger equation which may be implemented on a quantum computer. In Sec.~\ref{object_ins_env} we describe how obstacles can be efficiently implemented in the difference operators. In Sec.~\ref{euler} we derive quantum circuits modeling sound propagation and scattering described by the linearized Euler equations for both conservative and non-conservative cases. We prove statements about the complexity of the resulting algorithm in Sec.~\ref{comp_analais}. In Sec.~\ref{experiment} we present the results of numerical simulations for a number of obstacles, comparing them to a classical FDM approach. We conclude and provide outlook in Sec.~\ref{conclus}. Appendices provide details of the circuit constructions and of the mathematical proofs, as well as further numerical examples.

\begin{figure*}[t]
    \centering
    \includegraphics[width=\textwidth]{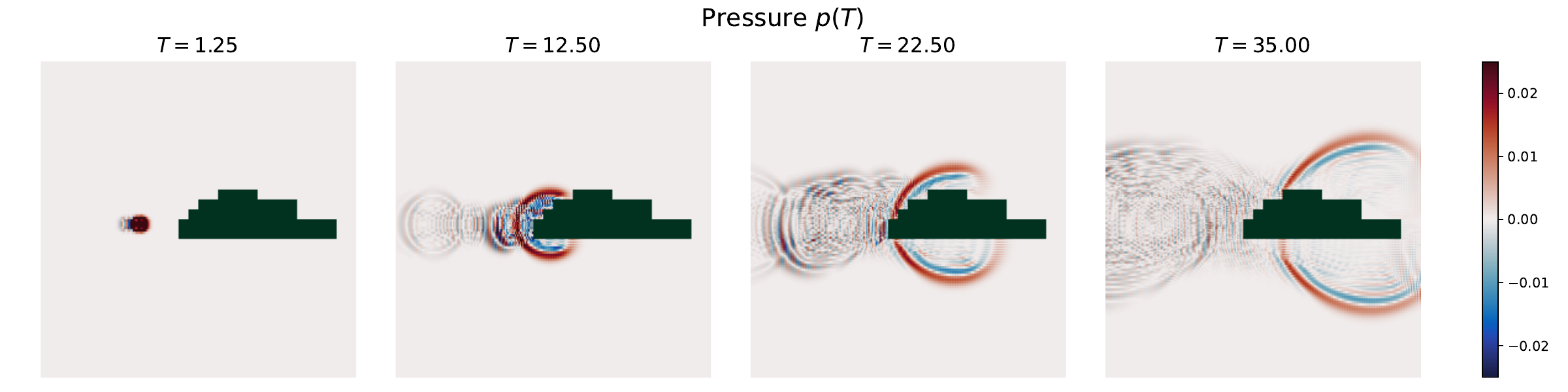}
    \caption{A quantum time-evolution simulation of the 2D linearized Euler equations (LEE) describing acoustic wave propagation in the presence of a constant background fluid flow, and an impenetrable airfoil-like obstacle. Here we plot the deviation of pressure from the equilibrium. The results were obtained by simulating the quantum circuits derived in this work, using $n=9$ qubits and a Trotter step of $\tau=0.05$.}
    \label{fig:airfoil}
\end{figure*}

\section{Differential equations as trotterized evolutions  \label{prelim}}


We begin by reviewing how some PDEs can be mapped to a quantum Hamiltonian evolution problem on a digital quantum computer by using explicit quantum circuit implementation of discrete difference operators. We largely follow the exposition and conventions of Ref.~\cite{toyota2024pde}. 

We consider a system of linear PDEs 
\begin{equation}\label{eq:pde_general}
    \frac{ {\rm d} f}{ {\rm d}t} = A f,
\end{equation}
where $f=(f_1,\dots,f_s)$ is a vector function of the $(x_1,\ldots,x_d)$ spatial coordinates, and $A$ is $s\times s$ matrix containing the physical parameters of the system and spatial derivatives $\partial/\partial {x_i}, \partial^2/\partial {x_i}\partial {x_j},\ldots$. 

When $H\coloneqq iA$ is a Hermitian operator, that is, $H^\dagger = H$,  Eq.~\ref{eq:pde_general} can be rewritten as Schr\"odinger equation
\begin{equation}\label{eq:shrodinger_eq}
\frac{ {\rm d} f}{ {\rm d}t} = -i H f,
\end{equation}
where $H$ plays a role of the Hamiltonian. Solving the equation can then, in theory, be obtained by simply time-evolving an initial condition $f(t=0)$ encoded in the amplitudes of a quantum state with the evolution operator $\exp(-iHt)$.

Performing such a procedure on a digital quantum computer is in practice more subtle. First, to encode the equation in a quantum state the PDE Eq.~\ref{eq:pde_general} needs to be discretized. This entails substituting all partial derivatives by suitably chosen difference operators, yielding a difference equation governed by the operator $D$, a discretized version of $A$. It is important to emphasize that it is $D$ which has to give rise to a hermitian $H$, and whether this is possible depends on the equation itself (including parameter values), the boundary conditions, and the discretization scheme chosen. Incorporating complex boundary conditions while preserving hermiticity, in particular, is of key practical importance in simulations and is our main focus.

Second, after mapping, the re-cast Hamiltonian evolution needs to be simulated on a digital QPU to a prescribed accuracy, which requires deriving quantum circuits implementing this goal. Whatever the simulation method chosen, the resulting circuit complexity in terms of the elementary gates of the QPU architecture must be shown not to invalidate the putative quantum advantage of the whole approach. We focus on the Trotter evolution and show explicitly that circuits implementing complex conditions, such as obstacles in the fluid flow, can be implemented with only polynomial gate complexity.

\subsection{Difference operators}\label{sec:differenceops}

Let us review the mapping procedure in more detail, focusing on, without loss of generality, the example of a scalar function in 1D (\emph{i.e.}~$s=d=1$ in Eq.~\ref{eq:pde_general}).

We first discretize the 1D support of $f\equiv f_1$ using $N=2^n$ equidistant points $x^k$, labeled by integers $k$, with a lattice constant $l$. We encode the values of $f$ at these points in the amplitudes of an $n$-qubit quantum state $|u\rangle$:
\begin{equation}\label{eq:vector_scalar_field}
    \vert u\rangle=\sum_{k=0}^{N-1} u_k\vert k\rangle, \mbox{\ \ \ \ with \ \ } u_k \coloneqq \frac{f(x^k)}{\sqrt{\sum_i |f(x^i)|^2}},
\end{equation}
where $\vert k\rangle \equiv |k^0k^1\ldots k^n\rangle$ are the computational basis states of the $n$-qubit system and $k^0k^1\ldots k^n$ is the binary expansion of $k$; the amplitudes of $|u\rangle$ are normalized. Note that while preparing a general state on a quantum computer can be exponentially complex, important classes, \emph{e.g.}~when $f$ is an arbitrary smooth function, can be prepared efficiently with linearly deep circuits using tensor network methods \cite{distr_loading}.

The first order derivative $\partial / \partial x$ can be replaced with a discrete forward, backward, or central difference operator $D^+$, $D^{-}$ and $D^{\pm}$. The action of these operators on the vector of discretized function values is given by:
\begin{align}\label{eq:D_pm_def}
    &(D^+ u)_k = \frac{u_{k+1}-u_{k}}{l},\nonumber\\
    &(D^- u)_k = \frac{u_{k}-u_{k-1}}{l},\nonumber\\
    &(D^{\pm} u)_k = \frac{u_{k+1}-u_{k-1}}{2l}.
\end{align}
To write down an explicit representation of the difference operators in the $n$-qubit system in terms of quantum gates we will first need the following shift operators:
\begin{align}\label{eq:shift_operators}
    S^-=\sum_{k=1}^{N-1}\vert k-1\rangle\langle k\vert =\sum_{j=1}^n{\rm I}^{\otimes(n-j)}\otimes\sigma_{01}\otimes\sigma_{10}^{\otimes(j-1)},\nonumber\\
     S^+=\sum_{k=1}^{N-1}\vert k\rangle\langle k-1\vert =\sum_{j=1}^n{\rm I}^{\otimes(n-j)}\otimes\sigma_{10}\otimes\sigma_{01}^{\otimes(j-1)},
\end{align}
where $\mathrm{I}$ is a 1-qubit identity and the 1-qubit operators $\sigma_{\alpha\beta}$ are given by the following $2\times 2$ matrices:
\begin{align}\label{eq:basis_matrices}
  & \sigma_{01} = \vert 0\rangle\langle 1\vert = \begin{pmatrix}
   0 & 1 \\
   0 & 0 
 \end{pmatrix},\quad
 \sigma_{10} = \vert 1\rangle \langle 0 \vert = \begin{pmatrix}
   0 & 0 \\
   1 & 0 
 \end{pmatrix},\nonumber\\
  &\sigma_{00} = \vert 0\rangle \langle 0 \vert = \begin{pmatrix}
   1 & 0 \\
   0 & 0 
 \end{pmatrix},\quad
  \sigma_{11} = \vert 1\rangle \langle 1 \vert = \begin{pmatrix}
   0 & 0 \\
   0 & 1 
 \end{pmatrix}.
\end{align}
Using the above auxiliary operators difference operators can be written down. Note that their representation depends on the boundary conditions imposed. With zero Dirichlet conditions at the ends of a simply connected interval they assume the following form:
\begin{align}\label{eq:D_pm_DBC}
    &D^+_D=\frac{1}{l}\left(S^--{\rm I}^{\otimes n}\right),\nonumber\\
    &D^-_D=\frac{1}{l}\left({\rm I}^{\otimes n}-S^+\right)\nonumber\\
    &D^{\pm}_D=\frac{1}{2l}\left(S^--S^+\right).
\end{align}
They are easily seen, from their explicit $N\times N$ matrix representations, to reproduce  Eqs.~\ref{eq:D_pm_def}:
\begin{align}\label{eq:D_pm}
    &D^{+}_{D} = \frac{1}{l} \begin{pmatrix}
   -1 & \ 1 & 0 & \cdots \\
   0 & -1 & \ 1 & \cdots \\
   0 & 0 & -1 & \cdots \\
   \vdots  & \vdots  & \vdots &  \ddots 
 \end{pmatrix},\nonumber\\
   &D^{-}_{D} = \frac{1}{l} \begin{pmatrix}
   \ 1 & 0 & 0 & \cdots \\
   -1 & \ 1 &  0 & \cdots \\
   0 & -1 & \ 1 & \cdots \\
   \vdots  & \vdots  & \vdots &  \ddots 
 \end{pmatrix},\nonumber\\
    &D^{\pm}_{D} = \frac{1}{2l} \begin{pmatrix}
   0 & \ 1 & 0 & \cdots \\
   -1 & 0 & \ 1 & \cdots \\
   0 & -1 & 0 & \cdots \\
   \vdots  & \vdots  & \vdots &  \ddots 
 \end{pmatrix}.
\end{align}
In what follows we shall focus on the Dirichlet conditions (dropping the ``D" subscript), and will derive the explicit representation of the difference operators in the more complicated and practically relevant case of impermeable obstacles within the function domain. Periodic boundary conditions of $D^\pm$ operator can be written by adding $(\sigma_{10}^{\otimes n} - \sigma_{01}^{\otimes n})/2l$ to Eq.~\ref{eq:D_pm}. Zero Neumann boundary conditions are implemented by nullifying the first and last rows of this matrix.

The generalization to a $d$-dimensional domain is simple: each dimension is discretized independently. While different numbers of qubits or grid distance can be used for each direction, for simplicity we consider the case when $n$ qubits are used per spatial coordinate with the same grid $l$.
Then, the total state on $d\cdot n$ qubits reads
\begin{equation}\label{eq:vector_scalar_field_d_dim}
    \vert u\rangle=\sum_{k_1=0}^{N-1} \sum_{k_2=0}^{N-1}\ldots \sum_{k_d=0}^{N-1} u_{k_1,k_2,\ldots k_d}\vert k_1, k_2,\ldots k_d \rangle
\end{equation}
with binary expansions $k_i \equiv k_i^0k_i^1\ldots k_i^n$.
The difference operator acting in the $\alpha$ direction is defined as
\begin{equation}\label{eq:D_pm_dim}
    D_{\alpha} = {\rm I}^{\otimes(\alpha-1)n}\otimes D\otimes {\rm I}^{\otimes (d-\alpha)n},
\end{equation}
where $D$ can contain any boundary conditions or differentiation scheme. This form follows from the choice of the encoding $|u\rangle$.

Finally, if the function $f$ is a $k$-component one with $k>1$, then the individual components are encoded using at least $m$ ancillary qubits, where $m = \lceil \log_2k \rceil$. For instance, the two components of a function $f = (f_1,f_2)$ could be encoded in the states $|0\rangle  \otimes |u^{(1)}\rangle$ and $|1\rangle \otimes |u^{(2)}\rangle $, and the complete state normalized. In this work we focus on this encoding. Thus, a $k$-component function of $d$ dimensions with $2^n$ discretisation points in each of them requires a total of $dn+m$ qubits.

\begin{figure}[t]
    \centering
\includegraphics[width=\columnwidth]{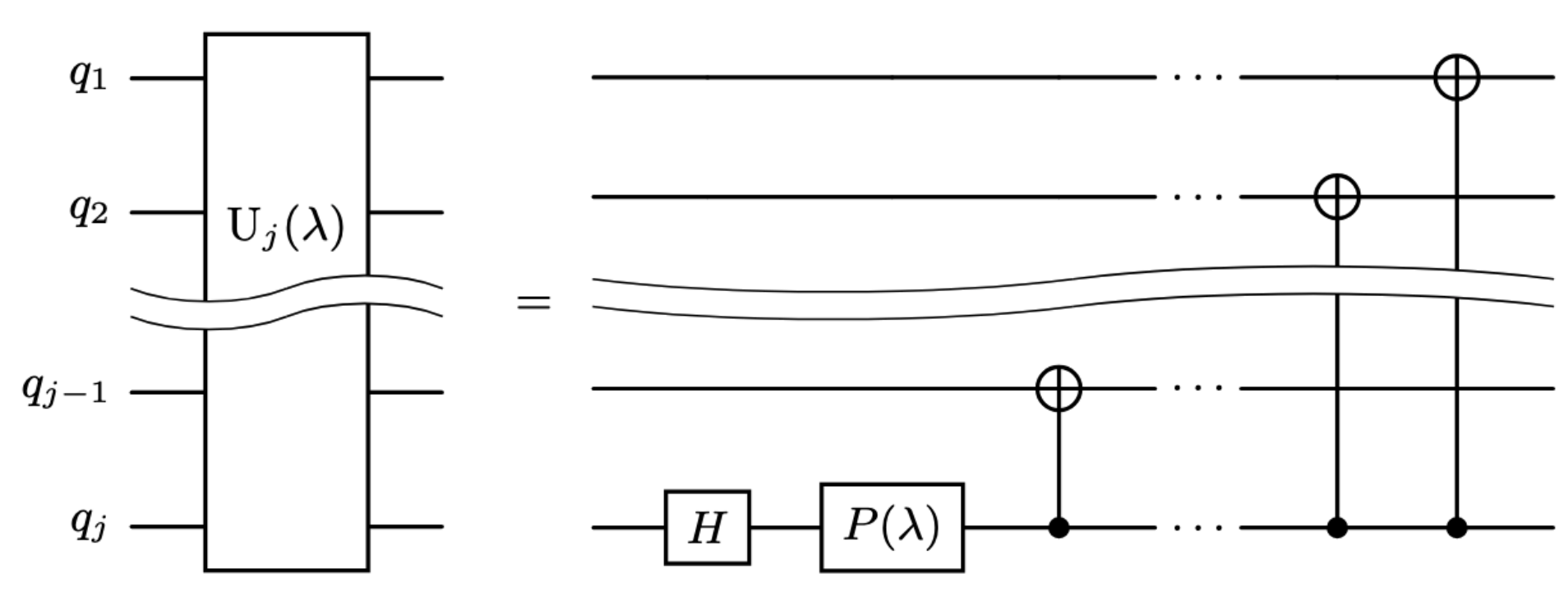}
    \caption{The quantum circuit for the ${\rm U}_j(\lambda)$ operator (see Eq.~\ref{eq:umatrix}) acting on the $q_{1},\dots,q_{j}$ register.}
    \label{fig:U_lambda_operator}
\end{figure}

\subsection{Induced evolution operators as quantum circuits}\label{sec:inducedops}

\begin{figure*}[!t]
    \centering
    \includegraphics[width=0.8\linewidth]{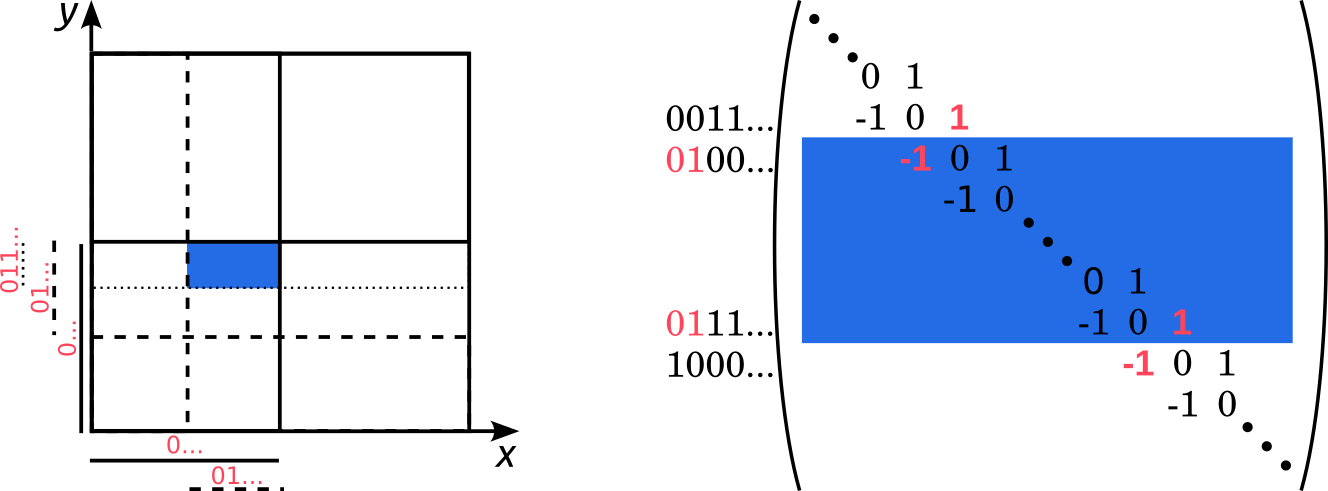}
    \caption{\textbf{(left)} A rectangular object placed inside a binary cell of the discretized domain. The cell is defined by the bit-string prefix $01$ in $x$-direction, and by $011$ in $y$-direction [see the main text]. \textbf{(right)} The matrix representation of the $2lD^\pm$ operator (here for the $x$-direction of the obstacle shown in the left panel), where rows which correspond to the points with the $x$-coordinate inside the obstacle are shaded in blue.}
    \label{fig:matrix_obstacle}
\end{figure*}

To actually implement the PDE on a quantum processor we need an explicit quantum circuit realization of the quantum evolution induced by the difference operator $\exp(-i(iD^\pm)\gamma)$.
A naive approach would be to directly express the shifts in terms of which the difference operator can be written (see Eq.~\ref{eq:shift_operators}) through Pauli strings -- this would, however, result in an exponentially growing number of terms. Instead, Ref.~\cite{toyota2024pde} derived a polynomially sized representation by diagonalizing the Hamiltonian terms in the modified Bell basis composed of states of the form $(|0\rangle|1\rangle^{\otimes j-1} \pm e^{-i\lambda} |1\rangle|0\rangle^{\otimes j-1})/\sqrt{2}$: applying the first order Trotter decomposition it can be shown that the evolution operator can be written as
\begin{align}\label{eq:bc1line}
    &e^{D^\pm \gamma} = e^{-i (iD^\pm) \gamma}\nonumber\\
    &= e^{-i \frac{i\gamma}{2l}\sum_{j=1}^{n}{\rm I}^{\otimes(n-j)}\otimes\left(\sigma_{01}\otimes\sigma_{10}^{\otimes (j-1)}-\sigma_{10}\otimes\sigma_{01}^{\otimes (j-1)}\right)} \\ \label{eq:bc1}
    &\approx \prod_{j=1}^{n} {\rm U}_j(-\pi/2) {\rm MCRZ}^{1,\dots, {j-1}}_{j}\left(\frac{\gamma}{l}\right) {\rm U}_j(-\pi/2)^\dagger.
\end{align}
Here, in the second line, before Trotterization, we used the definitions in Eqs.~\ref{eq:shift_operators} and \ref{eq:D_pm_DBC}; in the third line
${\rm MCRZ}_{j}^{1,\ldots j-1}\left(\gamma/l \right)$ is an $\exp{(-i{\rm Z}_j\gamma/(2l))}$ rotation of the $j$-th qubit (multi)-controlled on the $1,\ldots (j-1)$-th qubits, with ${\rm Z}$ the Pauli-Z matrix. 
The operator ${\rm U}_j(\lambda)$ mapping to the modified Bell basis is given by (see Fig.~\ref{fig:U_lambda_operator}):
\begin{align}\label{eq:umatrix}
  {\rm U}_j(\lambda)=\left(\prod_{m=1}^{j-1}{\rm CNOT}_m^j\right){\rm P}_j(\lambda){\rm H}_j,
\end{align}
where ${\rm CNOT}^{j}_m$ has the $j$-th qubit as the control and the $m$-th as the target, and ${\rm H}_j$, ${\rm P}_j(\lambda)$ are the Hadamard and phase shift operators on the $j$-th qubit:
\begin{equation}\label{eq:pshifted}
{\rm H}_j = \frac{1}{\sqrt{2}} \begin{pmatrix}
   1 & 1 \\
   1  & -1
 \end{pmatrix},\quad
    {\rm P}_j(\lambda) = \begin{pmatrix}
   1 & 0 \\
   0  & e^{i\lambda}
 \end{pmatrix}.
\end{equation}
If necessary, this operator can be implemented in logarithmic depth following the optimized implementations of the GHZ state \cite{Zhang_2022,Mooney_2021}.

A detailed derivation of Eq.~\ref{eq:bc1} and the upcoming CFD operators is presented in Appendices \ref{appegen}--\ref{appPeriodiccond}. Note, however, that the individual $j$ factors in the product in Eq.~\ref{eq:bc1} are generated by the corresponding $j$ terms in the sum in Eq.~\ref{eq:bc1line}, which, respectively, implement $2^{n-j}$ pairs of $\pm1$ above/below the diagonal. This intuition will prove useful later.

\section{Implementing finite difference operators with obstacles}\label{object_ins_env}

Having reviewed the implementation of the basic difference operators we now extend these results and demonstrate how complicated boundary conditions, in particular impermeable obstacles of arbitrary shapes, can be \emph{efficiently} introduced. This is of key practical importance in CFD simulations, when we are interested in, for instance, sound propagation around an object such as an airfoil or a car chassis moving in an airstream.

For pedagogical and clarity reasons we will now focus on two spatial dimensions (2D) denoted by $x$ and $y$; the extension to higher dimensions is trivial. We denote the qubits encoding the $x$ and $y$ coordinates by $q_{x,i}$ and $q_{y,i}$, respectively, where $1\leq i \leq n$. We will assume Dirichlet boundary conditions at the edges of the domain, in particular we set the value of the field to zero inside the area defining the obstacle. More general boundary conditions are implemented following the same prescription with the additional step of a LCU construction (see Appendix ~\ref{appDirichletNeumanncond}).

As introduced in Sec.~\ref{sec:differenceops} we use a binary discretization, where each point's coordinates $\alpha = x,y$ are labeled by bitstrings 
$b_\alpha \equiv b_{\alpha,1}\dots b_{\alpha,n}$ with $b_{\alpha,i}\in\{0,1\}$. We will sometimes find it convenient to think of the bitstring as a binary expansion of an integer $0\leq b_\alpha \leq 2^n-1$ which labels the consecutive points on the discretized coordinate axis. 
In such a grid we define a \emph{binary cell} labeled by a pair of strings  $(b_{x,1}\dots b_{x,k_x},\, b_{y,1}\dots b_{y,k_y})$ for $1\leq k_x,k_y \leq n$ to consist of all points $(b_x,b_y)$ whose coordinates have these cell strings as prefixes. The cell prefixes specify both its position and the area, which is easily seen to be $1/2^{k_x} \times 1/2^{k_y}$. An example of a binary cell, defined by the strings $(01,\, 011)$, is shown in Fig.~\ref{fig:matrix_obstacle}.

Let us begin by first implementing an elementary rectangular obstacle described by a single binary cell. In a classical numerical update scheme implementing a zero Dirichlet $D^\pm$ difference operator would entail setting all the rows corresponding to the obstacle coordinates in its matrix representation to zero, which from an algorithmic standpoint avoids any updates of the field values at these points. The matrix representation of $D^\pm$ for the $x$ coordinate is shown in Fig.~\ref{fig:matrix_obstacle}b, with the matrix elements to be put to zero in the shaded rectangle.

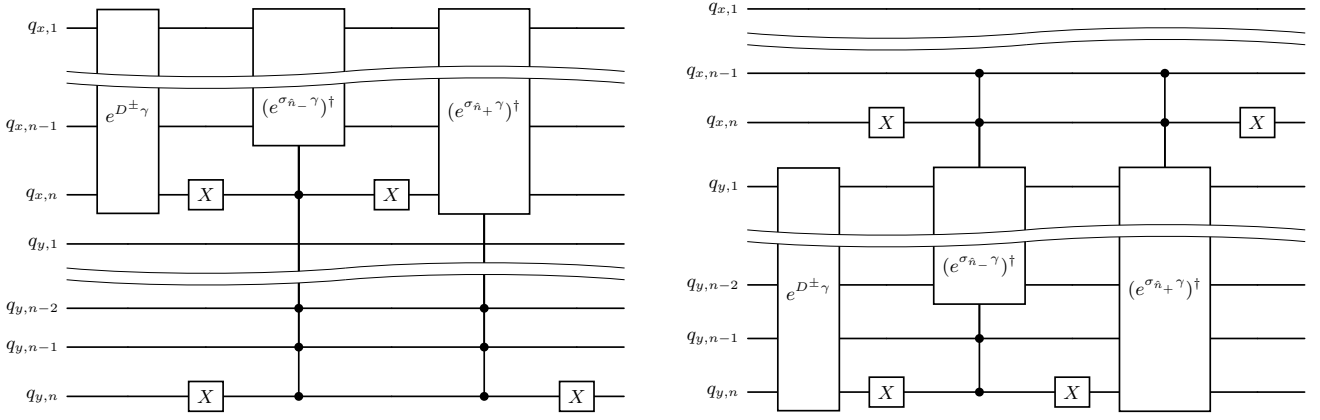
\begin{figure*}[!t]    
\begin{adjustbox}{width=\textwidth}
    \begin{quantikz}
\lstick{$q_{x,1}$}&\gate[4]{e^{D^\pm \gamma}}&&\gate[3,label
style={yshift=-0.5cm}]{(e^{\sigma_{\hat{n}_-} \gamma})^\dagger}&&\gate[4]{(e^{\sigma_{\hat{n}_+} \gamma})^\dagger}&&\\
\wave&&&&&&&&\\
\lstick{$q_{x,n-1}$}&&&&&&&\\
\lstick{$q_{x,n}$}&&\gate{X}&\ctrl{-3}&\gate{X}&&&\\
\lstick{$q_{y,1}$}&&&&&&&\\
\wave&&&&&&&&\\
\lstick{$q_{y,n-2}$}&&&\ctrl{-6}&&\ctrl{-6}&&\\
\lstick{$q_{y,n-1}$}&&&\ctrl{-7}&&\ctrl{-7}&&\\
\lstick{$q_{y,n}$}&&\gate{X}&\ctrl{-8}&&\ctrl{-8}&\gate{X}&
    \end{quantikz}
    \quad
        \begin{quantikz}
\lstick{$q_{x,1}$}&&&&&&&\\
\wave&&&&&&&&\\
\lstick{$q_{x,n-1}$}&&&\ctrl{2}&&\ctrl{2}&&\\
\lstick{$q_{x,n}$}&&\gate{X}&\ctrl{1}&&\ctrl{1}&\gate{X}&\\
\lstick{$q_{y,1}$}&\gate[5]{e^{D^\pm \gamma}}&&\gate[3,label
style={yshift=-0.5cm}]{(e^{\sigma_{\hat{n}_-} \gamma})^\dagger}&&\gate[5]{(e^{\sigma_{\hat{n}_+} \gamma})^\dagger}&&\\
\wave&&&&&&&&\\
\lstick{$q_{y,n-2}$}&&&&&&&\\
\lstick{$q_{y,n-1}$}&&&\ctrl{-3}&&&&\\
\lstick{$q_{y,n}$}&&\gate{X}&\ctrl{-4}&\gate{X}&&&
    \end{quantikz}
\caption{\textbf{(Left)} The quantum circuit implementing the difference operator (derivative) on the $x$ axis with an impenetrable object inside the 2D environment (with Dirichlet boundary conditions on its edges) defined by the bitstrings $b_{x,1}b_{x,2} = 01$ on the $x$ coordinate and by $b_{y,1}b_{y,2}b_{y,3} = 011$ on the $y$ coordinate. Here the $D^\pm$ blocks implement the obstacle-free evolution and the prefix-controlled operators (see the main text) remove the matrix elements corresponding to the obstacles. \textbf{(Right)} The corresponding circuit for the difference operator on the $y$ axis. Note that the``little-endian" convention is used, that is the left-most bits in the bit-strings correspond to the qubits with the largest index.}
\label{fig:D_operator_obstacle}
\end{adjustbox}
\end{figure*}

The quantum evolution implementation, however, is of a different nature, since all of the amplitudes are updated simultaneously. Rather than trying to directly implement an evolution ``skipping" all of the selected updates, which would result in a highly complex circuit, our approach will be different.
First, we implement the evolution under the full $D^\pm$ operator (see Eq.~\ref{eq:bc1}), and subsequently we cancel back some of the updates with an evolution under an auxiliary operator. Naively one may think that for this purpose a quantum circuit needs to be constructed, which removes all of the elements put to zero in the classical scheme, but in practice it suffices to construct an evolution which only sets to zero four entries in $D^\pm$ (marked red in Fig.~\ref{fig:matrix_obstacle}b), which greatly simplifies the circuit construction. Indeed, due to the choice of boundary conditions and field values inside the obstacle, putting the two marked matrix entries outside the shaded area to zero does not affect the numerical scheme, since the field values of interest are zero.

We will perform this construction explicitly, without loss of generality for the $x$ coordinate. Let the position of the obstacle on the $x$ axis be specified by the prefix $b_{x,1}\dots b_{x,k_x}$. It is easily seen that the matrix elements of the $D^\pm$ operator (for the $x$ coordinate) corresponding to the two entries marked in red in the top left of Fig.~\ref{fig:matrix_obstacle}b are given by:

\begin{align}\label{eq:sigma_minus}
    &|(b_{x,1}\dots b_{x,k_x}-1)111\dots \rangle\langle b_{x,1}\dots b_{x,k_x} 000\dots| \nonumber\\
    &- |b_{x,1}\dots b_{x,k_x}000\dots \rangle\langle (b_{x,1}\dots b_{x,k_x}-1)111\dots|,
\end{align}
while the bottom-right pair matrix elements are given by:
\begin{align}\label{eq:sigma_plus}
    &|b_{x,1}\dots b_{x,k_x}111\dots\rangle\langle(b_{x,1}\dots b_{x,k_x}+1)000\dots|\nonumber\\
    &- |(b_{x,1}\dots b_{x,k_x}+1)000\dots\rangle\langle b_{x,1}\dots b_{x,k_x}111\dots|,
\end{align}
where the $\pm 1$ additions in the prefixes denote bit strings obtained from increasing by one the integer, whose binary representation the original prefix was. In other words: we change the row/column index in the matrix by one, and take its binary representation, as dictated by the tridiagonal structure of the matrix $D^\pm$.

Let $p_x^-$ be the size of largest common prefix of $b_{x,1}\dots b_{x,k_x}$ and $b_{x,1}\dots b_{x,k_x}-1$, and $p_x^+$ the largest common prefix of $b_{x,1}\dots b_{x,k_x}$ and $b_{x,1}\dots b_{x,k_x}+1$ with $0\leq p_x^-,p_x^+<k_x$. Then it is easy to verify that
\begin{align}\label{eq:sigma_minus_prefix}
    &|(b_{x,1}\dots b_{x,k_x}-1)111\dots \rangle\langle b_{x,1}\dots b_{x,k_x} 000\dots| \nonumber\\
    &- |b_{x,1}\dots b_{x,k_x}000\dots \rangle\langle (b_{x,1}\dots b_{x,k_x}-1)111\dots| = \nonumber\\
    & = |b_{x,1}\dots b_{x,p_x^-}\rangle\langle b_{x,1}\dots b_{x,p_x^-}|\otimes \nonumber\\
    &\quad \left(|(0111\dots \rangle\langle 1000\dots|- |1000\dots \rangle\langle 0111\dots|\right) \equiv \nonumber \\
    &\equiv |b_{x,1}\dots b_{x,p_x^-}\rangle\langle b_{x,1}\dots b_{x,p_x^-}|\otimes 2l\cdot \sigma_{\hat{n}_-},
\end{align}
where the operator $\sigma_{\hat{n}}$ appearing in the last line with $\hat{n}:= \hat{n}_-$ (here: $\hat{n}_- = n-p_x^-$) comprises the suffix part acting on $\hat{n}$ qubits and is defined as:

\begin{align}\label{eq:hatn}
\sigma_{\hat{n}} \equiv \frac{1}{2l}(|0111\dots\rangle\langle1000\dots| - |1000\dots\rangle\langle0111\dots|).
\end{align}
Similarly it is easily seen that the operator in Eq.~\ref{eq:sigma_plus} can now be expressed as (with $\hat{n}_+ = n - p_x^+$):
\begin{equation}\label{eq:sigma_plus_prefix}
|b_{x,1}\dots b_{x,p_x^+}\rangle\langle b_{x,1}\dots b_{x,p_x^+}|\otimes 2l\cdot \sigma_{\hat{n}_+}.
\end{equation}
Eqs. \ref{eq:sigma_minus_prefix} and \ref{eq:sigma_plus_prefix} are valid for all prefixes except the edge cases at the boundary of the domain where they are all-0 or all-1. However, in this case the corresponding matrix elements are already absent in the full operator and there is no need for the correction.

With the above observations and definitions the difference operator without the matrix elements in Eqs.~\ref{eq:sigma_minus}, \ref{eq:sigma_plus} (see Fig.~\ref{fig:matrix_obstacle}) can be written as:
\begin{align}\label{eq:dpmsubs}
    D^\pm &- |b_{x,1}\dots b_{x,p_x^-}\rangle\langle b_{x,1}\dots b_{x,p_x^-}| \otimes \sigma_{\hat{n}_-}  \\ \nonumber
    &-  |b_{x,1}\dots b_{x,p_x^+}\rangle\langle b_{x,1}\dots b_{x,p_x^+}| \otimes \sigma_{\hat{n}_+}.
\end{align}

To implement the evolution generated by the the operator Eq.~\ref{eq:dpmsubs} we use the first order Trotter decomposition. First, we implement the evolution $\exp{(D^\pm \gamma)}$ according to the original difference operator (see Eq.~\ref{eq:bc1}). For the remaining ``obstacle" part, observe that the explicit matrix form  of the operator $\sigma_{\hat{n}}$ in Eq.~\ref{eq:hatn} consists of a single $\pm1$ pair above/below diagonal, \emph{i.e.}~has the same form as the $j = \hat{n}$ term in Eq.~\ref{eq:bc1line}. The evolution it generates can thus be written as:
\begin{align}\label{eq:sigma_exp}
    &e^{\sigma_{\hat{n}} \gamma} = e^{-i (i\sigma_{\hat{n}}) \gamma} = e^{-i \frac{i}{2l}\left(\sigma_{01}\otimes\sigma_{10}^{\otimes (\hat{n}-1)}-\sigma_{10}\otimes\sigma_{01}^{\otimes (\hat{n}-1)}\right)}\nonumber\\
    &= {\rm U}_{\hat{n}}(-\pi/2) {\rm MCRZ}^{ 1,\dots, \hat{n}-1}_{\hat{n}}\left(\frac{\gamma}{l}\right) {\rm U}_{\hat{n}} (-\pi/2)^\dagger.
\end{align}
Furthermore, the complete obstacle terms in Eq.~\ref{eq:dpmsubs} have the form of a tensor product of the $\sigma_{\hat{n}}$ operator and a projector on a specific prefix bit string, and thus the evolution according to them is obtained by taking a multi-controlled version of Eq.~\ref{eq:sigma_exp}, with the control selecting precisely the qubit states corresponding to the bit string required. Note, that here for clarity we have only considered the $x$-coordinate, but the obstacle operators have to also include the controls on the bit strings defining the extent of the obstacle in the other spatial directions; this is shown schematically in Fig.~\ref{fig:D_operator_obstacle}. We also remark that in this work we use the ``little-endian" convention between bits and qubits: the left-most bits of the bit-string correspond to qubits with the largest indices.

As an illustration of the above construction in a 2D example, consider again the obstacle depicted in Fig.~\ref{fig:matrix_obstacle}, defined by the string $b_{x,1}b_{x,2} = 01$ on the $x$ coordinate and by $b_{y,1}b_{y,2}b_{y,3} = 011$ on the $y$ coordinate. To construct the $x$ derivative with obstacles we first compute the common prefixes which have lengths $p_x^-=1$ and $p_x^+=0$. The resulting circuit is shown in Fig.~\ref{fig:D_operator_obstacle}: first, the full $D^\pm$ block, and then the two obstacle terms, controlled on the common prefixes in the $q_x$ register, and also on the full $011$ bit string in the $q_y$ register. The difference operator for the $y$ derivative is obtained analogously (\emph{i.e.}~controlled on the prefixes in the $y$ direction and the full 01 bitstring in the $x$ direction).

Though the above procedure provides a blueprint for implementing in the flow an obstacle of arbitrary shape (composed of multiple binary cells of different sizes) its practical implementation can be further improved to reduce both the theoretical and practical errors. For the former we note that a direct Trotter approximation of the evolution generated by the operator in Eq.~\ref{eq:dpmsubs} would violate the boundary conditions on the obstacles, due to the non-commutation of the complete obstacle and the $D^\pm$ terms. This, however, can be prevented entirely. As shown in Appendix \ref{appObstacle} the generator in Eq.~\ref{eq:dpmsubs} can be re-arranged so that the obstacle terms involving the $\sigma_{\hat{n}}$ generator are grouped with the corresponding $j=\hat{n}$ terms in $D^\pm$, with which they commute, even in their controlled version. This allows to implement the obstacle \emph{exactly}, without incurring any additional Trotter error with respect to the obstacle free evolution and thus ensures that the obstacles are indeed perfectly impermeable.

Furthermore, the above re-arrangement can be combined with additional circuit-level simplifications. Particularly the construction of multi-controlled and multi-target operators seen in Fig.~\ref{fig:D_operator_obstacle} can be made much more efficient, reducing the required circuit depth. While the gains obtained in this fashion amount to constant factors, these can have a very strong effect on the practical performance, especially on early QPUs currently (and in the near- to mid-term future) available. For details we refer to the Appendix \ref{appObstacle}.

\section{Linearized Euler equations with a mean flow \label{euler}}

Here we demonstrate the application of our methods to the quantum simulation of linearized Euler equations (LEE) in the presence of a background fluid flow, and in the presence of obstacles, for which we develop explicit circuit representations.

LEE, obtained by linearizing the complete Euler equations around a mean flow, allow  the simulation of \emph{i.a.}~acoustic wave propagation in aerospace and automotive design simulations, and are thus of key practical importance. From the point of view of quantum simulation a key challenge is that the presence of a mean flow generically hinders the mapping of LEE to a Hermitian operator described Sec.~\ref{prelim}. Here we show that for specific parameter values corresponding to the conservative regime this is still possible in a direct fashion. Furthermore, for the general non-conservative case, we derive a LCU implementation, whose building blocks are precisely of the form constructed for the conservative regime.

For concreteness let us now consider LEE with a background flow in the $x$ direction:
\begin{equation}\label{eq:pde_euler_full}
    \begin{cases}
      \frac{\partial p}{\partial t} = -\bar{\rho} c^2 \left( \frac{\partial u}{\partial x} + \frac{\partial v}{\partial y} \right) - \bar{u} \frac{\partial p}{\partial x},\\
      \frac{\partial u}{\partial t} = -\frac{1}{\bar{\rho}} \frac{\partial p}{\partial x} - \bar{u} \frac{\partial u}{\partial x},\\
      \frac{\partial v}{\partial t} = -\frac{1}{\bar{\rho}} \frac{\partial p}{\partial y} - \bar{u} \frac{\partial v}{\partial x},
    \end{cases}
\end{equation}
where $p$ is the deviation from the mean pressure (\emph{i.e.}~can be negative), $u$ and $v$ are the velocity components in the $x$ and $y$ directions, respectively, $\bar{\rho}$ is a static density of the environment, $\bar{u}$ is a speed of flow in $x$ direction and
$c$ is the speed of sound in the environment. 

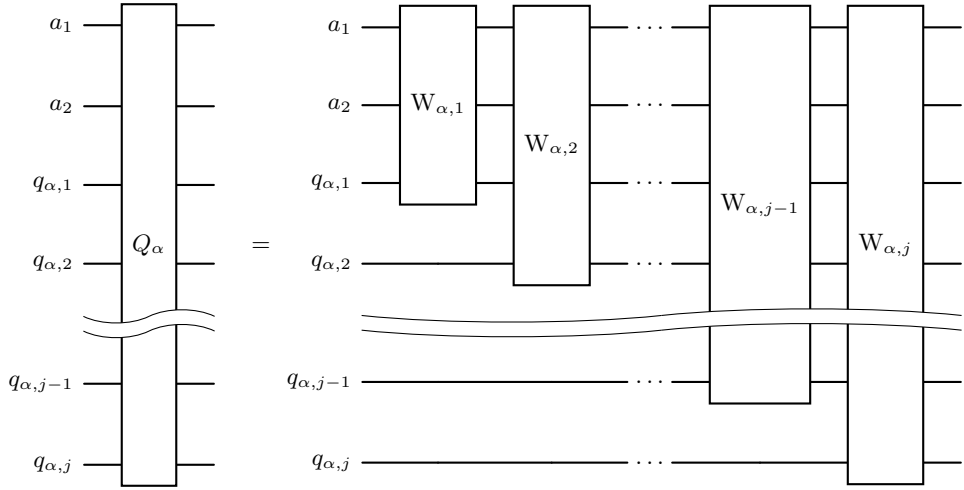
\begin{figure*}[ht!]
    \centering
    \begin{quantikz}[row sep=0.55cm]
\lstick{$a_1$}&\gate[7,label
style={yshift=0.0cm}]{Q_{\alpha}}&\ghost{X}\\
\lstick{$a_2$}&&\ghost{X}\\
\lstick{$q_{\alpha,1}$}&&\ghost{X}\\
\lstick{$q_{\alpha,2}$}&&\ghost{X}\\
\wave&&&&&&&\\
\lstick{$q_{\alpha,j-1}$}&&\ghost{X}\\
\lstick{$q_{\alpha,j}$}&&\ghost{X}
    \end{quantikz}\hspace{-2.cm}=
    \begin{quantikz}[row sep=0.5cm]
\lstick{$a_1$}&\gate[3,label style={yshift=0.0cm}]{{\rm W}_{\alpha,1}}&\gate[4,label style={yshift=0.0cm}]{{\rm W}_{\alpha,2}}& \ \ldots\ &\gate[6,label style={yshift=0.0cm}]{{\rm W}_{\alpha,j-1}}&\gate[7,label style={yshift=0.0cm}]{{\rm W}_{\alpha,j}}&\ghost{X}\\
\lstick{$a_2$}&&&\ \ldots\ &&&\ghost{X}\\
\lstick{$q_{\alpha,1}$}&&&\ \ldots\ &&&\ghost{X}\\
\lstick{$q_{\alpha,2}$}&&&\ \ldots\ &&&\ghost{X}\\
\wave&&&&&&&\\
\lstick{$q_{\alpha,j-1}$}&&&\ \ldots\ &&&\ghost{X}\\
\lstick{$q_{\alpha,j}$}&&&\ \ldots\ &&&\ghost{X}
    \end{quantikz}
    \caption{The quantum circuit implementing the $Q_{\alpha}$ time evolution   Eq.~\ref{eq:evolution_op2} for the $\alpha=x,y$ spatial direction.}
    \label{fig:Q_alpha_operator}
\end{figure*}

We represent the solution at each point as a four-component vector $f=(p,u,v,0)$, which requires  $m=2$ ancilla qubits to encode (see discussion in Sec.~\ref{sec:differenceops}); we denote them by $a_{1}$ and $a_2$. 
The zero component in the solution vector appears as we require it to be a power of $2$ in the qubit encoding.
With this LEE can be represented in the form of Eq.~\ref{eq:pde_general}, with the matrix $A$ given by:
\begin{equation}\label{eq:euler_evolution}
    A =
    \begin{pmatrix}
        -\bar{u} \frac{\partial}{\partial x} & -\bar{\rho}c^2 \frac{\partial}{\partial x} & -\bar{\rho}c^2 \frac{\partial}{\partial y} & 0\\
        -\frac{1}{\bar{\rho}} \frac{\partial}{\partial x} & -\bar{u} \frac{\partial}{\partial x} & 0 & 0 \\
        -\frac{1}{\bar{\rho}} \frac{\partial}{\partial y} & 0 & -\bar{u} \frac{\partial}{\partial x} & 0 \\
        0 & 0 & 0 & -\bar{u} \frac{\partial}{\partial x}
    \end{pmatrix}.
\end{equation}
Note the appearance of the derivative on the diagonal in the auxiliary dimension introduced. While it has no effect on the dynamics (the corresponding entry in the solution is always zero), it simplifies the circuit construction.

\subsection{The conservative regime}

It is easy to see that in the special case when the speed of sound is $c = 1 / \bar{\rho}$ the matrix in Eq.~\ref{eq:euler_evolution} is symmetric.
This case corresponds to the regime, where the total energy (kinetic and pressure) is conserved exactly. When applying the $D^\pm$ difference operator prescription for discretization, the whole matrix $A$ becomes \emph{anti-symmetric}, and thus the problem can be directly rewritten as a Schr\"odinger equation, as in Eq.~\ref{eq:shrodinger_eq}, with the Hamiltonian given by:

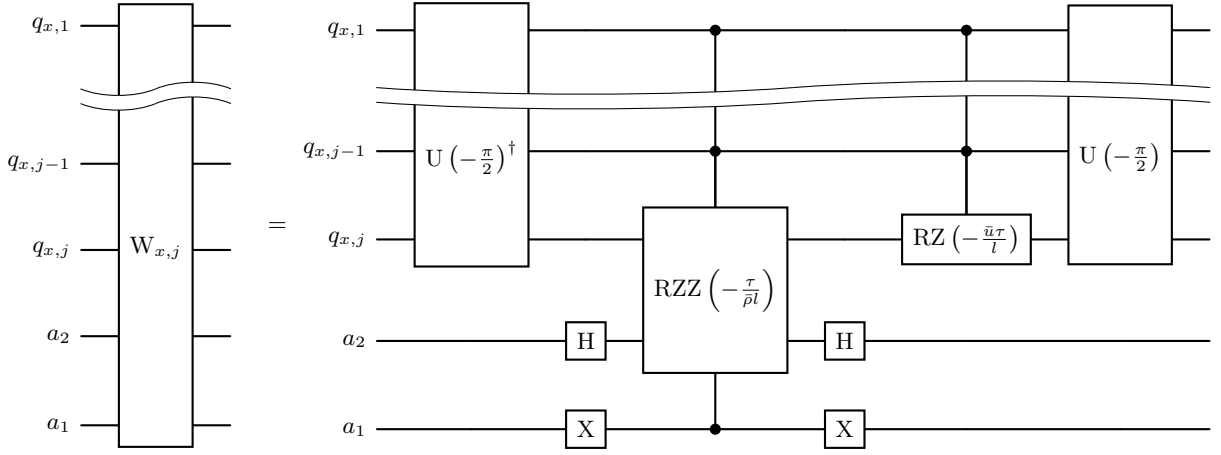
\begin{figure*}[ht!]
    \centering
    \begin{quantikz}[row sep=0.65cm]
\lstick{$q_{x,1}$}&\gate[6,label
style={yshift=-0.3cm}]{{\rm W}_{x,j}}&\ghost{X}\\
\wave&&&&&&&\\
\lstick{$q_{x,j-1}$}&&\ghost{X}\\
\lstick{$q_{x,j}$}&&\ghost{X}\\
\lstick{$a_2$}&&\ghost{X}\\
\lstick{$a_1$}&&\ghost{X}
    \end{quantikz}\hspace{-2.cm}=
    \begin{quantikz}[row sep=0.5cm]
\lstick{$q_{x,1}$}&\gate[4,label
style={yshift=-0.3cm}]{{\rm U}\left(-\frac{\pi}{2}\right)^\dagger}& & \ctrl{3}&&\ctrl{3}&\gate[4,label
style={yshift=-0.3cm}]{{\rm U}\left(-\frac{\pi}{2}\right)}&\ghost{X}\\
\wave&&&&&&&\\
\lstick{$q_{x,j-1}$}&&& \ctrl{1}&&\ctrl{1}&&\ghost{X}\\
\lstick{$q_{x,j}$}&&& \gate[2]{{\rm RZZ}\left(-\frac{\tau}{\bar{\rho}l}\right)}&&\gate{{\rm RZ}\left(-\frac{\bar{u}\tau}{l}\right)}&&\ghost{X}\\
\lstick{$a_2$}&&\gate{{\rm H}}& &\gate{{\rm H}}&&&\ghost{X}\\
\lstick{$a_1$}&&\gate{{\rm X}}& \ctrl{-2}&\gate{{\rm X}}&&&\ghost{X}
    \end{quantikz}
    \caption{The quantum circuit for the ${\rm W}_{x,j}$ operator acting on the $q_{x,1},\dots,q_{x,j}$ spatial and $a_1,a_2$ ancilla registers.}
    \label{fig:W_x_operator}
\end{figure*}

\begin{figure*}[ht!]
    \centering
    \begin{quantikz}[row sep=0.65cm]
\lstick{$q_{y,1}$}&\gate[6,label
style={yshift=-0.3cm}]{{\rm W}_{y,j}}&\ghost{X}\\
\wave&&&&&&&\\
\lstick{$q_{y,j-1}$}&&\ghost{X}\\
\lstick{$q_{y,j}$}&&\ghost{X}\\
\lstick{$a_1$}&&\ghost{X}\\
\lstick{$a_2$}&&\ghost{X}
    \end{quantikz}\hspace{-2.cm}=
    \begin{quantikz}[row sep=0.5cm]
\lstick{$q_{y,1}$}&\gate[4,label
style={yshift=-0.3cm}]{{\rm U}\left(-\frac{\pi}{2}\right)^\dagger}&&\ctrl{3}& &\gate[4,label
style={yshift=-0.3cm}]{{\rm U}\left(-\frac{\pi}{2}\right)}&\ghost{X}\\
\wave&&&&&&&\\
\lstick{$q_{y,j-1}$}&&& \ctrl{1}&&&\ghost{X}\\
\lstick{$q_{y,j}$}&&& \gate[2]{{\rm RZZ}\left(-\frac{\tau}{\bar{\rho}l}\right)}&&&\ghost{X}\\
\lstick{$a_1$}&&\gate{{\rm H}}& &\gate{{\rm H}}&&\ghost{X}\\
\lstick{$a_2$}&&\gate{{\rm X}}& \ctrl{-2}&\gate{{\rm X}}&&\ghost{X}
    \end{quantikz}
    \caption{The quantum circuit for the ${\rm W}_{y,j}$ operator acting on the $q_{y,1},\dots,q_{y,j}$ spatial and $a_1,a_2$ ancilla registers.}
    \label{fig:W_y_operator}
\end{figure*}
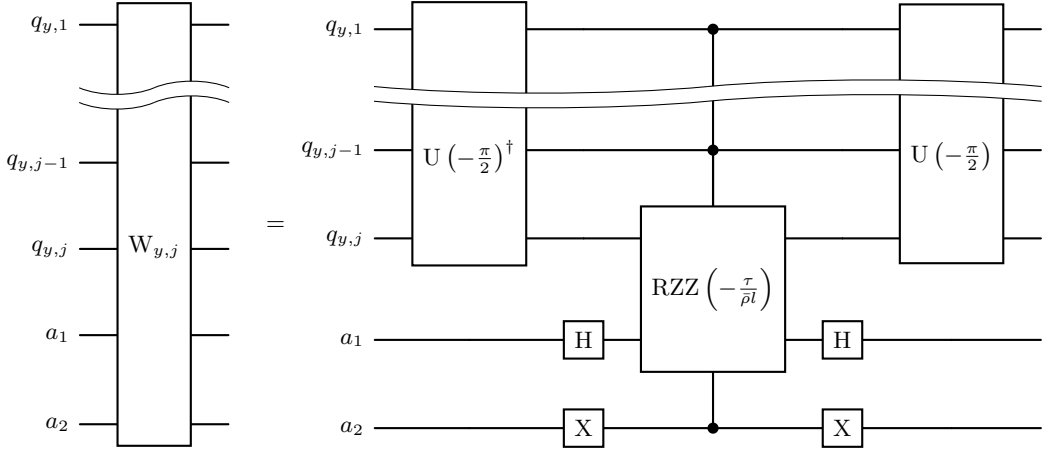

\begin{align}\label{eq:hamiltonian}
   H=&-i\left[\bar{u} {\rm I}^{a_1}\otimes {\rm I}^{a_2} \otimes D_x^\pm + \frac{1}{\bar{\rho}}\sigma_{00}^{a_1}\otimes{\rm X}^{a_2}\otimes D_x^\pm\right.\nonumber\\
   &\left.+ \frac{1}{\bar{\rho}}{\rm X}^{a_1}\otimes\sigma_{00}^{a_2}\otimes D_y^\pm\right].
\end{align}
Here and later, if necessary, we indicate the qubits an operator acts on as an upper index; $\sigma_{00}$ is defined in Eq.~\ref{eq:basis_matrices}. We further group the Hamiltonian terms related to the $x$ and $y$ derivatives, respectively, and write them out explicitly using the shift operators introduced in Eq.~\ref{eq:shift_operators}:
\begin{align}\label{eq:hamiltonian2}
   H= H_x+H_y=\sum_{j=1}^nH_{x,j}+\sum_{j=1}^nH_{y,j},
\end{align}
where
\begin{align}
&H_{x,j}=-\frac{i}{2l} \left(\bar{u} {\rm I}^{a_1}\otimes {\rm I}^{a_2} + \frac{1}{\bar{\rho}}\sigma_{00}^{a_1}\otimes{\rm X}^{a_2} \right)\nonumber\\
   &\otimes {\rm I}^{\otimes(n-j)}\otimes(\sigma_{01}\otimes\sigma_{10}^{\otimes(j-1)}-\sigma_{10}\otimes\sigma_{01}^{\otimes(j-1)})\otimes {\rm I}^{\otimes n},\nonumber\\
   &H_{y,j}=-\frac{i}{2l\bar{\rho}}{\rm X}^{a_1}\otimes\sigma_{00}^{a_2}\nonumber\\
   &\otimes {\rm I}^{\otimes n} \otimes {\rm I}^{\otimes(n-j)}\otimes(\sigma_{01}\otimes\sigma_{10}^{\otimes(j-1)}-\sigma_{10}\otimes\sigma_{01}^{\otimes(j-1)}).\nonumber
\end{align}

 The time-evolution ${\rm V}(\tau)$ over a small interval $\tau$ can thus be written using first-order Trotterization as:
\begin{align}\label{eq:evolution_op}
    V(\tau)=e^{-iH\tau}\approx e^{-iH_y\tau} e^{-iH_x\tau} \equiv Q_y(\tau)Q_x(\tau).
\end{align}
The non-commuting ``per-direction" evolution operators $Q_x$ and $Q_y$ can be further written as products:

\begin{align}\label{eq:evolution_op2}
    Q_x=\prod_{j=1}^n {\rm W}_{x,j},\quad Q_y=\prod_{j=1}^n{\rm W}_{y,j}.
\end{align}
The operators ${\rm W}_{x,j}$ and ${\rm W}_{y,j}$ act on the two ancillas and $j$ spatial discretization qubits (see Fig.~\ref{fig:Q_alpha_operator}). Their circuit representations are shown in Figs.~\ref{fig:W_x_operator}, \ref{fig:W_y_operator} and the corresponding explicit expressions are given below: 
\begin{widetext}
\begin{align}\label{eq:evolution_operator4}
&{\rm W}_{x,j}={\rm I}^{a_1}\otimes {\rm I}^{a_2}\otimes {\rm I}^{\otimes(n-j)}\otimes \left( {\rm U}_j(-\pi/2){\rm MCRZ}_{j}^{1,\ldots j-1}\left(-\bar{u}\tau/l\right){\rm U}_j(-\pi/2)^{\dagger}\right) \otimes {\rm I}^{\otimes n}\nonumber\\
&\times {\rm I}^{\otimes(n-j)}\otimes \left({\rm H}^{a_2}\otimes {\rm X}^{a_1}\otimes {\rm U}_j(-\pi/2)\right){\rm MCRZZ}_{a_2,j}^{a_1,1,\ldots j-1}\left(-\tau/(\bar{\rho}l) \right)\left({\rm U}_j(-\pi/2)^{\dagger}\otimes{\rm X}^{a_1}\otimes {\rm H}^{a_2}\right)\otimes{\rm I}^{\otimes n},\nonumber\\
&{\rm W}_{y,j}={\rm I}^{\otimes n} \otimes {\rm I}^{\otimes(n-j)}\otimes \left({\rm H}^{a_1}\otimes {\rm X}^{a_2}\otimes {\rm U}_j(-\pi/2)\right){\rm MCRZZ}_{a_1,j}^{a_2,1,\ldots j-1}\left(-\tau/(\bar{\rho}l) \right)\left({\rm U}_j(-\pi/2)^{\dagger}\otimes {\rm H}^{a_1}\otimes {\rm X}^{a_2}\right),
\end{align}
\end{widetext}
where the multi-controlled $Z$-rotation operator {\rm MCRZ} was introduced in Sec.~\ref{sec:inducedops} and where, analogously, ${\rm MCRZZ}_{a_1(a_2),j}^{a_2(a_1),1,\ldots j-1}\left(-\tau/(\bar{\rho}l) \right)$ provides an $\exp{(i{\rm Z}_{a_1(a_2)}{\rm Z}_j \tau/(2\bar{\rho}l))}$ Ising rotation of qubits $a_1$ (respectively: $a_2$) and $j$, controlled on the qubits $a_2$ (respectively: $a_1$) and $1,\ldots, j-1$. A detailed derivation of the operator ${\rm V}(\tau)$ is presented in Appendix ~\ref{appegen}.

The total evolution over the time period $T=s\tau$ is obtained in the first-order Trotter scheme as $f(T)\approx \left[{\rm V}(\tau)\right]^sf(0)$, where the quantum state encoding the initial condition should be properly normalized $\|f(0)\|=1$. The operators $Q_x$ and $Q_y$ can also be used to implement higher-order Trotter-Suzuki evolutions \cite{Suzuki1991}. Note also that the complexity of quantum circuit preparing the initial state is highly dependent on the state itself, and may in general be considerable.

This completes the derivation of the evolution circuit for the LEE Eq.~\ref{eq:pde_euler_full} with a constant background flow in the conservative regime. It can now be combined with the obstacle circuits derived in Sec.~\ref{object_ins_env} and App.~\ref{appObstacle} (and App.~\ref{appDirichletNeumanncond} for zero Neumann and mixed boundary conditions) to allow the quantum simulation of acoustic phenomena in non-trivial settings. The only difference is that the obstacle terms are now implemented for the $D^\pm$ difference operators in the Hamiltonian Eq.~\ref{eq:hamiltonian}. We emphasize that while the conservative regime is a special case, the obtained circuits are the building blocks also in the general non-conservative case, where the evolution cannot be directly mapped to a Hamiltonian evolution (see below and the Appendices~\ref{appNonlinear} and \ref{appDirichletNeumanncond}).  

Finally, we remark that the circuits depicted in Figs.~\ref{fig:W_x_operator}, \ref{fig:W_y_operator} contain multi-controlled and multi-target rotation MCRZZ operators. In Appendix~\ref{appedpdm} we show that using the $D^{\pm}$ \emph{and} the $D^{+}$, $D^{-}$ difference operators simultaneously the evolution operator construction can be simplified, in particular avoiding multi-target gates. We further derive the circuit for the periodic boundary conditions case in Appendix~\ref{appPeriodiccond}. 

\subsection{Beyond the conservative regime}

Implementing more physically relevant and realistic flows such as the LEE problem away from the special conservative point $c=1/\bar{\rho}$ requires additional techniques and computational resources, as the system cannot be directly written as a Hamiltonian evolution. The same is true of more general boundary conditions, such a Neumann or mixed ones.

As we show in more detail in Appendices~\ref{appNonlinear} and \ref{appDirichletNeumanncond} both of these can be addressed with the linear combination of unitaries (LCU) approach. The key idea is to split the operator into its real and imaginary parts, \emph{e.g.}~$D^{\pm}_{DN} = A_1 + iA_2$, both of which are Hermitian
\begin{equation}\label{eq^hermitizationofA_main}
    A_1 = \frac{A + A^\dagger}{2}, \quad\quad\quad A_2 = \frac{A - A^\dagger}{2i},
\end{equation}
and to use Trotter approximation to evolve separately under them. The unitary part originating from $A_2$ can be implemented directly as described above, so the problem reduces to implementing the purely imaginary-time evolution generated by $A_1$. This can be done using an integral representation \cite{An_2023}
\begin{equation}
e^{D^{\pm}_{DN}\gamma} = e^{-\frac{\gamma}{4l}{\rm I}}\int_{-\infty}^{+\infty} \frac{1}{\pi(1+k^2)}e^{i\left(k\left(A_1+\frac{1}{4l}{\rm I}\right)+A_2\right)\gamma} {\rm d}k,
\label{eq:hermeitianoperatorsDNcond2_main}
\end{equation}
which is then approximated by a finite sum over $M$ discrete values of $k$. Such sum of $M$ terms, each of which is a unitary evolution, is then in turn implemented using the LCU method, which requires $\log_2(M)$ auxiliary qubits. We refer to the Appendices for a detailed derivation of the necessary quantum circuits and a discussion of their complexity and error.

\section{Algorithm complexity \label{comp_analais}}

We briefly describe the resources required to implement the methods we discussed. We use $n$ qubits to represent the discretized grid of $2^n$ points in each direction (we focus on two spatial dimensions here). The most important practical metric in analysing the circuit complexity is the time increment $\tau=T/s$ ($s$ is the number of Trotter steps), the number of two-qubit (typically CNOT) gates; the number of one-qubit gates is of a similar order.

\begin{figure*}[t!]
    \centering
    \includegraphics[width=0.85\textwidth]{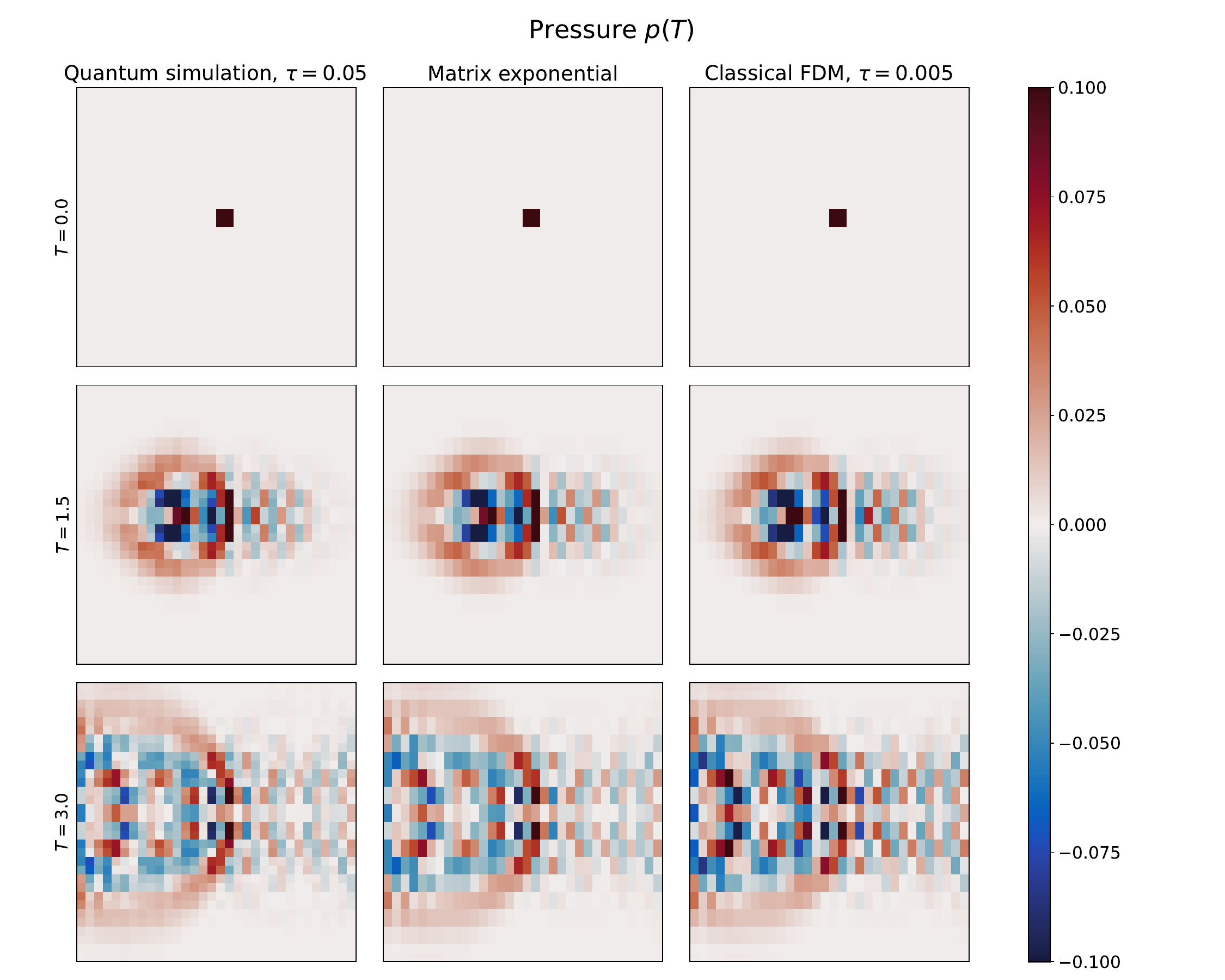}
    \caption{Numerical comparison of the pressure component of the solution at the different time $T$ between quantum simulation (this work), exact matrix exponential solution and the classical finite difference method (FDM) for $n=5$.}
    \label{fig:euler_ideal_comparison}
\end{figure*}

\begin{lemma}\label{lemma1} The evolution operator $\exp{(-iH\tau)}$ with the Hamiltonian $H$ given in Eq.~\ref{eq:hamiltonian} and a time increment $\tau$ can be approximated in the first-order Lie-Trotter-Suzuki decomposition by the unitary operator ${\rm V}(\tau)$ Eq.~\ref{eq:evolution_op} with an approximation error (in the sense of operator spectral norm) upper bounded bounded by:

 \begin{align}\label{eq:upper_error}
     \lVert \exp{(-iH\tau)}& - {\rm V}(\tau)\rVert \nonumber\\
     \leq& \left[\left(\frac{\bar{u}}{2}\right)^2+2\left(\frac{1}{2\bar{\rho}}\right)^2 +
     \frac{\bar{u}}{2\bar{\rho}} \right]\frac{\tau^2(n-1)}{2l^2}\nonumber\\
     &+\left(\frac{1}{2\bar{\rho}}\right)^2\frac{\tau^2 n^2}{2l^2}.
 \end{align}
\end{lemma}
For the proof of Lemma~\ref{lemma1} we refer to the Appendix~\ref{appCommutators}.
We also show that the gate complexity of the unitary operator ${\rm V}(\tau)$ is quadratic in the number of qubits, as formalized by the following lemma:
\begin{lemma}\label{lemma2}
The evolution operator ${\rm V}(\tau)$ Eq.~\ref{eq:evolution_op} depicted in Figs.~\ref{fig:Q_alpha_operator}, \ref{fig:W_x_operator},  \ref{fig:W_y_operator} can be implemented with $\mathcal{O}(n^2)$ ${\rm CNOT}$ gates in all-to-all qubit connectivity.
\end{lemma}
\begin{proof}
    The operator ${\rm V}(\tau)$ consists of multiple ${\rm U}_j(\lambda)$ operators and multi-controlled rotations. We show that the implementation of both groups does not exceed quadratic complexity.
    First, note that by Eq.~\ref{eq:umatrix} a single ${\rm U}_j(\lambda)$ contains $j-1$ CNOTs, and ${\rm V}(\tau)$ contains a linear number of operators ${\rm W}$, and thus of ${\rm U}_j(\lambda)$, as per Figs.~\ref{fig:Q_alpha_operator}, \ref{fig:W_x_operator},  \ref{fig:W_y_operator}. In total $2n(n-1)$ ${\rm CNOT}$ gates are enough to implement all ${\rm U}_j(\lambda)$ and ${\rm U}_j(\lambda)^\dagger$ operators.

    Since each ${\rm W_{x,j}}$ contains a MCRZZ and a MCRZ gate (see Fig.~\ref{fig:W_x_operator}) and ${\rm W_{y,j}}$ an MCRZZ gate (Fig.~\ref{fig:W_y_operator}) the evolution ${\rm V}(\tau)$ contains $n$ ${\rm MCRZ}$ gates and $2n$ ${\rm MCRZZ}$ gates in total.
    The latter can be decomposed into two ${\rm MCRZ}$ gates with an additional control (see Fig.~\ref{fig:mcrzz_decompozition}). While decompositions of general multi-controlled gates may require quadratic number of CNOTs \cite{silva2023lineardecompositionapproximatemulticontrolled,Barenco_1995,da_Silva_2022}, in special cases, particualrly with ${\rm RZ}$ as the controlled unitary, efficient decompositions linear in the number of ${\rm CNOT}$ gates are known \cite{vale2023decompositionmulticontrolledspecialunitary,Iten_2016}. This provides an upper boundary $\mathcal{O}(n^2)$ for the complexity of the Trotter step, completing the proof.
\end{proof}

Similarly, it can be easily shown that the total amount of single-qubit gates is $\mathcal{O}(n^2)$ per Trotter step. The exact number is dependent on the decomposition used and device topology and gate set. We remark that in sparse topologies the ${\rm SWAP}$ gates can increase the gate complexity by a factor of $n$. For the all-to-all topology, using the decompositions in Ref.~\cite{vale2023decompositionmulticontrolledspecialunitary}, at most $42n^2-34n+34$ ${\rm CNOT}$ gates are needed to implement $V(\tau)$ for $n\geq 3$.

The total simulation complexity over time interval $T = \tau s$ can now estimated:

\begin{lemma}\label{lemma3}
The CNOT cost of simulating the evolution operator $\exp{(-iHT)}$ with the 2D Hamiltonian Eq.~\ref{eq:hamiltonian} with $2n+2$ qubits on a grid with size $l$ to an additive error $\varepsilon>0$ is: 
\begin{equation}\label{eq:totalcomplexity}
    \mathcal{O}\left( \frac{T^2 n^4}{l^2 \varepsilon} \right)
\end{equation}
\end{lemma}
\begin{proof}
    Using the estimate from Lemma~\ref{lemma1} we have that to achieve the additive error $\varepsilon$ over the time interval $T=\tau s$ we need $s$ to be of order:
    \begin{equation}
        s = \mathcal{O}\left( \frac{T^2 n^2}{l^2 \varepsilon} \right).
    \end{equation}
    At each step we implement an operator ${\rm V}(\tau)$ with a quadratic CNOT complexity, as given by Lemma~\ref{lemma2}. Combining the two results we obtain Eq.~\ref{eq:totalcomplexity}.
\end{proof}

We have thus obtained that even with obstacles the total simulation cost remains polynomial in the number of qubits, \emph{i.e.}~it provides a possible path to an exponential advantage over classical schemes explicitly updating a $2^n\times 2^n$ spatial grid, at least when $l = 1/poly(n)$. This holds for an obstacle implemented to a finite resolution, \emph{i.e.} with a constant amount of binary cells (see App.~\ref{appObstacle}). The complexity and error due to the LCU construction in the non-conservative case, which only brings a polynomial overhead, is analysed in Appendix~\ref{appDirichletNeumanncond}, but only brings polynomial overhead.

We emphasize though that the above analysis should not be confused for the statement that a direct end-to-end execution of such algorithm automatically provides quantum advantage. As is well known quantum state preparation and readout is generically the computational bottleneck, and this is also the case with Hamiltonian evolution-based methods \cite{toyoizumi2023hamiltoniansimulationusingquantum,PhysRevX.13.041041,Jin2022,jin2023quantumsimulationdiscretelinear,toyota2024pde,satohiroukinaoki2024}. While initial states can often be prepared using shallow quantum circuits (see \emph{e.g.} \cite{distr_loading}), the readout is typically a larger problem, as is the case even for the seminal HHL algorithm \cite{PhysRevLett.103.150502}. This separate problem is typically addressed by only computing select observables \cite{PhysRevLett.103.150502}, employing, where applicable, polynomially efficient tomography methods \cite{cramer2010efficient,song2024incompressiblenavierstokessolvenoisy,termanova2024tensorquantumprogramming,williams2024addressingreadoutproblemquantum,fang2025polyqromorthogonalpolynomialbasedquantumreducedorder}, or embedding the algorithm in a larger quantum computation, avoiding a complete readout of the final simulation state \cite{stein2023exponential}.

\section{Numerical simulation and experiments \label{experiment}}

Here we demonstrate the results of executing quantum circuits for LEE, derived using the methods described above, on an ideal simulator. We show both examples with Dirichlet, as well as more physically meaningful free slip rigid wall conditions, and employ different geometries and source terms, to demonstrate generality. We validate their correctness by comparing them to the classical finite difference method (FDM) using the forward Euler scheme, as the Trotterized evolution quantum approach is also a forward scheme. We further compare the results to a direct (classical) matrix exponential solution $e^{At}f(0)$, equivalent to the other two in the limit $\tau \rightarrow 0$. 
In all of the plots we show the pressure component $p$ of the solution (note that in the LEE $p$ is actually the deviation from the mean pressure, and can assume both positive and negative values).

If not stated otherwise, for the initial conditions we consider a single high pressure point-source of noise. We put $u(0)=v(0)=0$ everywhere and set $p(0)=0$ in all points except for a small $2 \times 2$ square in the center, where the pressure is set to $p=0.5$. Such initial conditions are easily prepared with several $\rm{X}$ and $\rm{H}$ gates. For more complex initial conditions shallow circuits may be prepared using Tensor Network approaches \cite{distr_loading}.

In Fig.~\ref{fig:euler_ideal_comparison} we plot the numerical simulation results for our method, and compare it to the forward Euler FDM scheme and exact matrix exponentiation for $n=5$ qubits. We use $c=\bar{\rho}=1$, $\bar{u}=-1$ and $l=0.25$. The quantum and classical solutions closely agree. Note though that in the quantum case we used a time step $\tau=0.05$, while in the classical numerical scheme we required a finer step of $\tau=0.005$ to maintain numerical stability (see Fig.~\ref{fig:l2_error_unstable}). This is because, unlike the classical one, the quantum solution only contains the Trotter error, which, furthermore, in some cases can even be independent of system size and simulation time \cite{heyl2019quantum}.
The solution error (in the pressure distribution) measured as the $L_2$ distance between either the quantum or classical FDM scheme and the exact matrix exponential is shown in Fig.~\ref{fig:l2_error}.

\begin{figure}
    \centering
    \includegraphics[width=\textwidth]{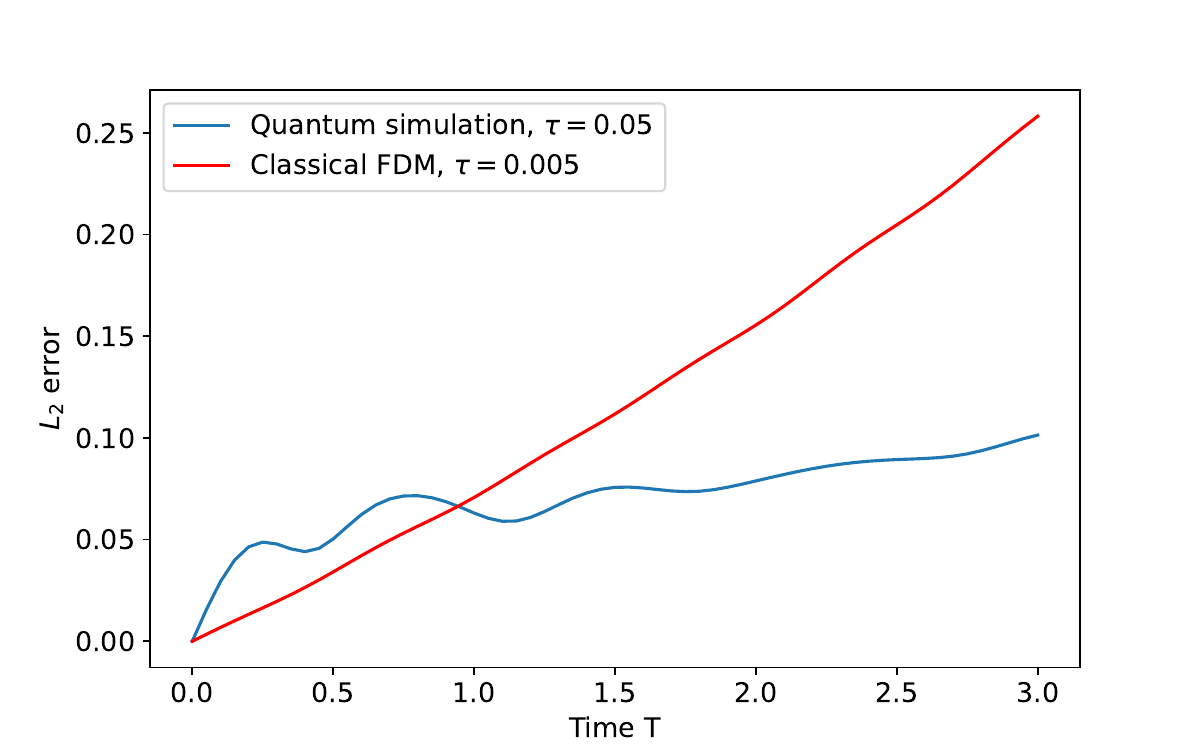}
    \caption{The $L_2$-error of the quantum and classical FDM solutions (both of which are forward schemes) for the pressure $p(T)$ in Fig.\ref{fig:euler_ideal_comparison} with respect to the exact matrix exponential solution.
    }
    \label{fig:l2_error}
\end{figure}

Further, in Fig.~\ref{fig:airfoil} we show the results of executing the quantum circuit for the 2D pressure wave propagating on a constant background constant flow and scattering from an impenetrable object inside the environment. The object in this case is a low-scale representation of an airfoil. To simplify the circuit construction we assume that the airfoil does not disturb the uniform background flow, but forms a reflective boundary for the acoustic wave. Nonuniform background flows can be implemented using methods described in Ref.\cite{satohiroukinaoki2024}. We set $c=\bar{\rho}=1$, $\bar{u}=2$, $\tau=0.05$, $l=0.5$. To avoid reflection from the boundaries of the domain (where we also imposed Dirichlet boundary conditions) we perform the simulation on a large $512\times512$ grid ($n=9$) and observe the wave propagating before it reaches the edges of the environment.


The above examples only employed Dirichlet boundary conditions, which cannot account for many physically interesting settings. Our construction is not limited to such setting, however. In Fig.~\ref{fig:euler_pipe_free_slip} we demonstrate a simulation of the LEE using the physically motivated free slip rigid wall boundary conditions, which implementation uses the more general LCU approach. Towards this end, let $D^\pm_D$ be the finite difference operator for Dirichlet boundary conditions as before, and $D^{\pm}_N$ be the central finite difference operator with zero Neumann boundary conditions. That is, the same matrix as in Eq.~\ref{eq:D_pm} but with zeroed first and last rows (see also Eq.~\ref{eq:DpmdirichletNeuman}). Then, to implement LEE with free slip rigid walls boundary conditions we substitute the following operators into the system Eq.\ref{eq:euler_evolution}:
\begin{equation}\label{eq:equler_free_slip}
    \begin{pmatrix}
        -\bar{u} D^\pm_{N,x} & -\bar{\rho}c^2 D^\pm_{D,x} & -\bar{\rho}c^2 D^\pm_{D,y} & 0\\
        -\frac{1}{\bar{\rho}} D^\pm_{N,x} & -\bar{u} D^\pm_{D,x} & 0 & 0 \\
        -\frac{1}{\bar{\rho}} D^\pm_{N,y} & 0 & -\bar{u} D^\pm_{N,x} & 0 \\
        0 & 0 & 0 & -\bar{u} D^\pm_{D,x}
    \end{pmatrix}.
\end{equation}
This matrix operator is no longer anti-symmetric and does not correspond to a conservative regime. We implement the evolution following the recipe in the Appendices~\ref{appNonlinear} and~\ref{appDirichletNeumanncond}. We set $c=\bar{\rho}=1$, $\bar{u}=0.5$, $l=0.5$, $\tau=0.05$. Additionally to demonstrate that the proposed methods are not limited to a square environment and not constrained by the initial conditions, we create two sources with different intensities. We use 7 qubits for $x$-direction and 5 qubits for $y$-direction, creating a simulation in a pipe. For the LCU scheme, which implements non-conservative part of the evolution, we use 5 qubits. The results of the simulation are shown in Fig.~\ref{fig:euler_pipe_free_slip}, where we also implemented Dirichlet boundary conditions for the same geometry and initial state, for comparison. Note that in this experiment, since only reflections off the horizontal pipe walls are considered, there is no need to implement the Neumann boundary conditions for velocity $v$ (vertical walls are out of reach). As expected, both simulations match until the horizontal walls are reached.

Finally, to illustrate the flexibility of our method in implementing objects withing the flow, in Appendix~\ref{appExperiments} we show quantum simulations of the 2D acoustic wave propagation with circuits derived for arbitrarily complex shapes (see Fig.~\ref{fig:cat_ufo}).


\begin{figure*}[t!]
    \centering
    \includegraphics[width=\textwidth]{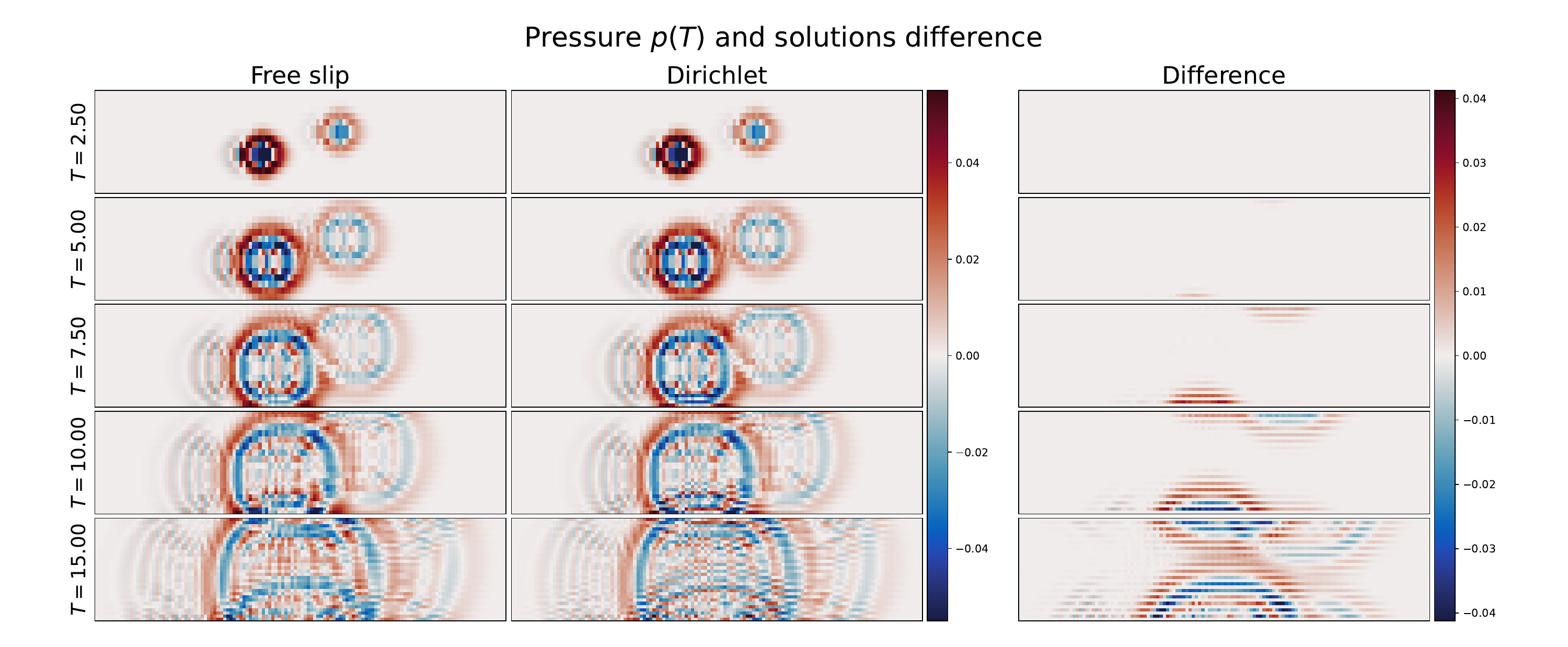}
    \caption{Quantum simulation of a 2D acoustic wave generated by two point sources of differing intensity propagating on top of a constant fluid flow in a pipe. We used $128\times32$ grid. Free slip boundary conditions are implemented \textbf{(left column)} with the LCU scheme, and compared to Dirichlet boundary conditions implemented with pure Hermitian operator \textbf{(middle column)}. Both solutions are compared \textbf{(right column)} and match till one of waves reaches the boundary of the pipe.}
    \label{fig:euler_pipe_free_slip}
\end{figure*}

\section{Conclusions \label{conclus}}

We investigated extending the recently introduced methods of solving PDEs via mapping them to a Hamiltonian evolution problem to realistic settings. Towards this goal we devised a method to generate explicit quantum circuit incorporating in the discretized derivative operators boundary conditions for arbitrarily complex shapes. This allows modeling phenomena such as acoustic wave scattering from airfoils. We showed that our construction does not change the asymptotic complexity of the circuit (in the non-LCU case) or comes with polynomial overheads (in the more general LCU case), as long as the obstacle is implemented to a finite resolution, preserving the exponential advantage of the free-space quantum method over classical solutions.

We further applied this technique to the problem of simulating the linearized Euler equations (LEE) Eq.~\ref{eq:pde_euler_full}. We demonstrated that in the conservative regime LEE with background flow ($c=1/\bar{\rho}$) can be directly implemented as a Hamiltonian evolution, and constructed explicit quantum circuits for it. For the non-conservative case, we gave a LCU \cite{An_2023,jin2025schrodingerizationbasedquantumalgorithms,guseynov2024explicitgateconstructionblockencoding} based implementation, using the conservative case circuits as building blocks. We further incorporated the boundary conditions corresponding to obstacles placed in the flow and simulated the quantum solution to acoustic wave scattering from them. All throughout we took care to describe optimized gate constructions, reducing both the theoretical as well practical implementation errors. We expect that the methods here described, as well as the detailed circuit constructions, will prove very useful to researchers investigating solving PDEs on the existing and upcoming quantum computers. 

\section*{Data Availability}

The quantum circuits used in the numerical experiments, as well as the classical and quantum simulation scripts are provided in a Zenodo repository (\url{https://doi.org/10.5281/zenodo.18232442}).

\bibliographystyle{unsrt}
\bibliography{haiqucfd}{}	

\medskip

\appendix

\section{Circuit derivation for the conservative system with central finite difference operators \label{appegen}}
\setcounter{equation}{0}
\renewcommand{\theequation}{A\arabic{equation}}

Here we explicitly derive the quantum circuit for the evolution described by the operator in Eq.~\ref{eq:euler_evolution} with $c = 1 / \bar{\rho}$. 
First,
we rewrite some of the factors appearing in the definition of the $H_{x,j}$ and $H_{y,j}$ terms in the Hamiltonian Eq.~\ref{eq:hamiltonian2} using the operators ${\rm U}_j(\lambda)$ given in Eq.~\ref{eq:umatrix}:
\begin{widetext}
\begin{align}\label{eq:sum_operator}
  &i{\rm I}^{\otimes(n-j)}\otimes(\sigma_{01}\otimes\sigma_{10}^{\otimes(j-1)}-\sigma_{10}\otimes\sigma_{01}^{\otimes(j-1)})=i{\rm I}^{\otimes(n-j)}\otimes\left(\vert 0\rangle\langle 1\vert \otimes \vert 1\rangle\langle 0\vert^{\otimes(j-1)}-\vert 1\rangle\langle 0\vert \otimes \vert 0\rangle\langle 1\vert^{\otimes(j-1)}\right)\nonumber\\
  &=i{\rm I}^{\otimes(n-j)}\otimes\left(\vert 0\rangle\otimes \vert 1\rangle^{\otimes (j-1)}\langle 1\vert \otimes \langle 0 \vert^{\otimes(j-1)}-\vert 1\rangle\otimes \vert 0\rangle^{\otimes (j-1)}\langle 0\vert \otimes \langle 1 \vert^{\otimes(j-1)}\right)\nonumber\\
  &={\rm I}^{\otimes(n-j)}\otimes\left(\frac{\vert 0\rangle\otimes\vert 1\rangle^{\otimes (j-1)}-i\vert 1\rangle\otimes\vert 0\rangle^{\otimes (j-1)} }{\sqrt{2}}\frac{\langle 0\vert\otimes\langle 1\vert^{\otimes (j-1)}+i\langle 1\vert\otimes\langle 0\vert^{\otimes (j-1)} }{\sqrt{2}}\right.\nonumber\\
  &\left.- \frac{\vert 0\rangle\otimes\vert 1\rangle^{\otimes (j-1)}+i\vert 1\rangle\otimes\vert 0\rangle^{\otimes (j-1)} }{\sqrt{2}}\frac{\langle 0\vert\otimes\langle 1\vert^{\otimes (j-1)}-i\langle 1\vert\otimes\langle 0\vert^{\otimes (j-1)} }{\sqrt{2}}\right)\nonumber\\
  &={\rm I}^{\otimes(n-j)}\otimes {\rm U}_j(-\pi/2)\left({\rm Z}\otimes \sigma_{11}^{\otimes (j-1)}\right){\rm U}_j(-\pi/2)^{\dagger}.
\end{align}
\end{widetext}
The effect of unitary matrix ${\rm U}_j(\lambda)$ (see Fig.~\ref{fig:U_lambda_operator} for the circuit representation) is to perform a basis change to the modified Bell basis so that:
\begin{align}
&{\rm U}_j(\lambda)\vert 0\rangle\otimes \vert 1\rangle^{\otimes (j-1)}\nonumber\\
&=\frac{1}{\sqrt{2}}\left(\vert 0\rangle\otimes\vert 1\rangle^{\otimes (j-1)}+e^{i\lambda}\vert 1\rangle\otimes\vert 0\rangle^{\otimes (j-1)} \right)\nonumber
\end{align}
and
\begin{align}
&{\rm U}_j(\lambda)\vert 1\rangle\otimes \vert 1\rangle^{\otimes (j-1)}\nonumber\\
&=\frac{1}{\sqrt{2}}\left(\vert 0\rangle\otimes\vert 1\rangle^{\otimes (j-1)}-e^{i\lambda}\vert 1\rangle\otimes\vert 0\rangle^{\otimes (j-1)} \right).\nonumber
\end{align}
Using Eq.~\ref{eq:sum_operator} we immediately obtain Eqs.~\ref{eq:bc1} and \ref{eq:sigma_exp}.

Continuing, the LEE Hamiltonian Eq.~\ref{eq:euler_evolution} can now be rewritten in the following form:
\begin{align}\label{eq:hamiltonian3}
H=H_x^{diag}+H_x^{off}+H_y^{off},
\end{align}
where
\begin{widetext}
\begin{align}
   &H_x^{diag}= -\frac{\bar{u}}{2l} {\rm I}^{a_1}\otimes {\rm I}^{a_2}\otimes\sum_{j=1}^n {\rm I}^{\otimes(n-j)}\otimes {\rm U}_j(-\pi/2)\left({\rm Z}\otimes \sigma_{11}^{\otimes (j-1)}\right){\rm U}_j(-\pi/2)^{\dagger}\otimes {\rm I}^{\otimes n},\label{eq:hamiltonian4}\\
   &H_x^{off}= -\frac{1}{2\bar{\rho}l}\sigma_{00}^{a_1}\otimes{\rm X}^{a_2} \otimes\sum_{j=1}^n {\rm I}^{\otimes(n-j)}\otimes {\rm U}_j(-\pi/2)\left({\rm Z}\otimes \sigma_{11}^{\otimes (j-1)}\right){\rm U}_j(-\pi/2)^{\dagger}\otimes {\rm I}^{\otimes n},\label{eq:hamiltonian5}\\
   &H_y^{off}=-\frac{1}{2\bar{\rho}l}{\rm X}^{a_1}\otimes\sigma_{00}^{a_2}\otimes {\rm I}^{\otimes n}\otimes \sum_{j=1}^n {\rm I}^{\otimes(n-j)}\otimes {\rm U}_j(-\pi/2)\left({\rm Z}\otimes \sigma_{11}^{\otimes (j-1)}\right){\rm U}_j(-\pi/2)^{\dagger}.\label{eq:hamiltonian6}
\end{align}
\end{widetext}
Above we split the Hamiltonian into a diagonal part $H_x^{diag}$ and two off-diagonal ones, corresponding to the derivatives $x$ and $y$, respectively.
Using the Lie-Trotter-Suzuki decomposition we obtain the first-order approximation the time evolution operator:
\begin{widetext}
\begin{align}\label{eq:evolution_operator}
&\exp{\left(-iH\tau\right)}\approx\exp{\left(-iH_x^{diag}\tau\right)}\exp{\left(-iH_x^{off}\tau\right)}\exp{\left(-iH_y^{off}\tau\right)}\nonumber\\
&\approx\prod_{j=1}^{n} \exp{\left(i \frac{\bar{u}\tau}{2l} {\rm I}^{a_1}\otimes {\rm I}^{a_2}\otimes {\rm I}^{\otimes(n-j)}\otimes {\rm U}_j(-\pi/2)\left({\rm Z}\otimes \sigma_{11}^{\otimes (j-1)}\right){\rm U}_j(-\pi/2)^{\dagger}\otimes {\rm I}^{\otimes n}\right)}\nonumber\\
&\times \prod_{j=1}^{n} \exp{\left(i \frac{\tau}{2\bar{\rho}l}\sigma_{00}^{a_1}\otimes{\rm X}^{a_2}\otimes {\rm I}^{\otimes(n-j)}\otimes {\rm U}_j(-\pi/2)\left({\rm Z}\otimes \sigma_{11}^{\otimes (j-1)}\right){\rm U}_j(-\pi/2)^{\dagger}\otimes {\rm I}^{\otimes n}\right)}\nonumber\\
&\times \prod_{j=1}^{n} \exp{\left(i \frac{\tau}{2\bar{\rho}l}{\rm X}^{a_1}\otimes\sigma_{00}^{a_2}\otimes {\rm I}^{\otimes n} \otimes {\rm I}^{\otimes(n-j)}\otimes {\rm U}_j(-\pi/2)\left({\rm Z}\otimes \sigma_{11}^{\otimes (j-1)}\right){\rm U}_j(-\pi/2)^{\dagger}\right)}.
\end{align}
\end{widetext}
Using the identities
\begin{align}
&\exp{\left(-i {\rm U} A {\rm U}^\dagger\right)}={\rm U}\exp{\left(-iA\right)}{\rm U^\dagger},\nonumber\\
&{\rm X}={\rm H}{\rm Z}{\rm H},\quad
\sigma_{00}={\rm X}\sigma_{11} {\rm X},
\end{align}
for a unitary matrix $U$, the above expression is rewritten as
\begin{widetext}
\begin{align}\label{eq:evolution_operator3}
&\exp{\left(-iH\tau\right)}\approx \prod_{j=1}^{n} {\rm I}^{a_1}\otimes {\rm I}^{a_2}\otimes {\rm I}^{\otimes(n-j)}\otimes {\rm U}_j(-\pi/2)\exp{\left(i \frac{\bar{u}\tau}{2l} {\rm Z}\otimes \sigma_{11}^{\otimes (j-1)}\right)}{\rm U}_j(-\pi/2)^{\dagger}\otimes {\rm I}^{\otimes n}\nonumber\\
&\times \prod_{j=1}^{n} {\rm I}^{\otimes(n-j)}\otimes \left({\rm X}^{a_1}\otimes {\rm H}^{a_2}\otimes {\rm U}_j(-\pi/2)\right)\exp{\left(i \frac{\tau}{2\bar{\rho}l}\sigma_{11}^{a_1}\otimes{\rm Z}^{a_2}\otimes {\rm Z}\otimes \sigma_{11}^{\otimes (j-1)}\right)}\left({\rm U}_j(-\pi/2)^{\dagger}\otimes{\rm X}^{a_1}\otimes {\rm H}^{a_2}\right)\otimes{\rm I}^{\otimes n}\nonumber\\
    &\times \prod_{j=1}^{n} {\rm I}^{\otimes n} \otimes {\rm I}^{\otimes(n-j)}\otimes \left({\rm H}^{a_1}\otimes {\rm X}^{a_2}\otimes {\rm U}_j(-\pi/2)\right)\exp{\left(i \frac{\tau}{2\bar{\rho}l}{\rm Z}^{a_1}\otimes\sigma_{11}^{a_2}\otimes {\rm Z}\otimes \sigma_{11}^{\otimes (j-1)}\right)}\left({\rm U}_j(-\pi/2)^{\dagger}\otimes {\rm H}^{a_1}\otimes {\rm X}^{a_2}\right).
\end{align}
\end{widetext}
We now observe that terms in the exponents are generators for multi-controlled ${\rm RZ}(\gamma)$ and ${\rm RZZ}(\gamma)$ rotation gates. Note that ${\rm RZ}(\gamma) = \exp(-i \gamma {\rm Z} / 2)$, therefore division by 2 in the central difference scheme is absorbed into the gate. With the above, finally, the circuit can be expressed in the form given in Eqs.~\ref{eq:evolution_op}, \ref{eq:evolution_op2}, \ref{eq:evolution_operator4}. In Fig.~\ref{fig:W_x_operator} we also grouped up some terms to reduce amount of $U_j(\lambda)$ gates.

\section{Circuit derivation for the conservative system with $D^{+}$ and $D^{-}$ difference operators \label{appedpdm}}
\setcounter{equation}{0}
\renewcommand{\theequation}{B\arabic{equation}}

In this appendix we derive the quantum circuit for the Euler equations with the background flow (in the conservative regime $c=1/\bar{\rho}$) with first order derivative operators. We use the $D^{\pm}$ difference operator to quantize derivatives of the diagonal part and $D^{+}$, $D^{-}$ difference operators to quantize derivatives of the off-diagonal part of the matrix Eq.~\ref{eq:euler_evolution}. The $D^{+}$ and $D^{-}$ are used for the super-diagonal and sub-diagonal terms, respectively. Then the evolution matrix remains Hermitian and can be expressed as:
\begin{align}\label{eq:hamiltonian_dp_dm}
   &\tilde{H}=-i\left[\bar{u} {\rm I}^{a_1}\otimes {\rm I}^{a_2} \otimes D^{\pm}\otimes {\rm I}^{\otimes n}\right.\\
   &\left.+ \frac{1}{\bar{\rho}}\left(\sigma_{00}^{a_1}\otimes\sigma_{01}^{a_2}\otimes D^+\otimes {\rm I}^{\otimes n}+\sigma_{00}^{a_1}\otimes\sigma_{10}^{a_2}\otimes D^-\otimes {\rm I}^{\otimes n}\right)\right.\nonumber\\
    &\left.+ \frac{1}{\bar{\rho}}\left(\sigma_{01}^{a_1}\otimes\sigma_{00}^{a_2}\otimes {\rm I}^{\otimes n}\otimes D^++\sigma_{10}^{a_1}\otimes\sigma_{00}^{a_2}\otimes {\rm I}^{\otimes n}\otimes D^-\right)\right].\nonumber
\end{align}
Note that the first line gives the diagonal elements with the derivatives w.r.t.~$x$ only, while the second and third lines contain off-diagonal elements with derivatives w.r.t.~to $x$ and $y$, respectively. Similarly to the previous section, the Hamiltonian Eq.~\ref{eq:hamiltonian_dp_dm} is also rewritten using the ${\rm U}_j(\lambda)$ operators defined in Eq.~\ref{eq:umatrix}. The first term with the $D^{\pm}$ operator was already derived in Eq.~\ref{eq:hamiltonian4}. Let us now rewrite the $x$ derivative off-diagonal terms. Substituting the representation of the $D^{+}$ and $D^-$ operators with Dirichlet boundary conditions Eq.~\ref{eq:D_pm_DBC} and the definitions of the shift operators Eq.~\ref{eq:shift_operators} into these terms we obtain:
\begin{widetext}
\begin{align}\label{eq:_hamiltonian_x_term}
   \tilde{H}^{off}_x=&-\frac{i}{\bar{\rho}}\left(\sigma_{00}^{a_1}\otimes\sigma_{01}^{a_2}\otimes D^+\otimes {\rm I}^{\otimes n}+\sigma_{00}^{a_1}\otimes\sigma_{10}^{a_2}\otimes D^-\otimes {\rm I}^{\otimes n}\right)\nonumber\\
   =&-\frac{i}{\bar{\rho}l}\left(\sigma_{00}^{a_1}\otimes\left(\sigma_{01}^{a_2}\otimes S^--\sigma_{10}^{a_2}\otimes S^{+} \right) -\sigma_{00}^{a_1}\otimes\left(\sigma_{01}^{a_2}-\sigma_{10}^{a_2}\right) \otimes {\rm I}^{\otimes n} \right)\otimes {\rm I}^{\otimes n}\nonumber\\
   =&-\frac{1}{\bar{\rho}l} \sum_{j=1}^n{\rm I}^{\otimes(n-j)}\otimes\left({\rm X}^{a_1}\otimes {\rm X}^j\otimes{\rm U}_{a_2}\left(-\frac{\pi}{2}\right)\right)\left(\sigma_{11}^{a_1}\otimes{\rm Z}^{a_2}\otimes \sigma_{11}^{\otimes j}\right)\left({\rm U}_{a_2}\left(-\frac{\pi}{2}\right)^{\dagger}\otimes {\rm X}^{j}\otimes{\rm X}^{a_1}\right)\otimes{\rm I}^{\otimes n}\nonumber\\
   &-\frac{1}{\bar{\rho}l}{\rm X}^{a_1}\left(\sigma_{11}^{a_1}\otimes {\rm Y^{a_2}}\right){\rm X^{a_1}}\otimes {\rm I}^{\otimes n} \otimes {\rm I}^{\otimes n},
\end{align}
\end{widetext}
where we used the identities
\begin{align}
\sigma_{01}^j=X^j\sigma_{10}^j X^j,\quad \sigma_{01}^{a_2}-\sigma_{10}^{a_2}=i{\rm Y}^{a_2}\nonumber
\end{align}
and Eq.~\ref{eq:sum_operator}. Here, an upper index $j$ refers to action on the $j$-th qubit in $q_x$ register and ${\rm Y}$ is a Pauli-Y matrix. With a small abuse of notation we also denoted by $U_{a_2}(\lambda)$ a $(j+1)$-qubit operator Eq.~\ref{eq:umatrix} acting on the first $j$ qubits and an ancillary qubit $a_2$.

For the $y$-derivatives term of the Hamiltonian Eq.~\ref{eq:hamiltonian_dp_dm} we obtain a similar result:
\begin{widetext}
\begin{align}\label{eq:_hamiltonian_y_term}
   \tilde{H}^{off}_y=&-\frac{i}{\bar{\rho}}\left(\sigma_{01}^{a_1}\otimes\sigma_{00}^{a_2}\otimes {\rm I}^{\otimes n}\otimes D^++\sigma_{10}^{a_1}\otimes\sigma_{00}^{a_2}\otimes {\rm I}^{\otimes n}\otimes D^-\right)\nonumber\\
   =&-\frac{1}{\bar{\rho}l}{\rm I}^{\otimes n}\otimes\sum_{j=1}^n{\rm I}^{\otimes(n-j)}\otimes\left({\rm X}^{a_2}\otimes {\rm X}^j\otimes {\rm U}_{a_1}\left(-\frac{\pi}{2}\right)\right)\left({\rm Z}^{a_1}\otimes \sigma_{11}^{a_2}\otimes \sigma_{11}^{\otimes j}\right)\left({\rm U}_{a_1}\left(-\frac{\pi}{2}\right)^{\dagger}\otimes {\rm X}^{j}\otimes {\rm X}^{a_2}\right)\nonumber\\
   &-\frac{1}{\bar{\rho}l}{\rm I}^{\otimes n}\otimes {\rm X}^{a_2}\left({\rm Y^{a_1}}\otimes\sigma_{11}^{a_2}\right){\rm X^{a_2}}\otimes {\rm I}^{\otimes n} .
\end{align}
\end{widetext}

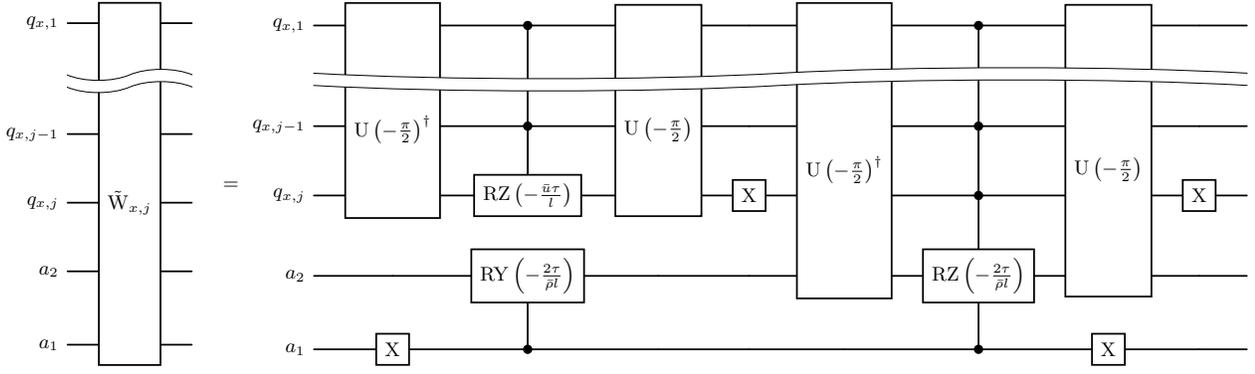
\begin{figure*}[ht]
   \centering
   \resizebox{\textwidth}{!}{%
   \begin{quantikz}[row sep=0.6cm]
\lstick{$q_{x,1}$}&\gate[6,label
style={yshift=-0.3cm}]{{\rm \tilde{W}}_{x,j}}&\ghost{X}\\
\wave&&&&&&&\\
\lstick{$q_{x,j-1}$}&&\ghost{X}\\
\lstick{$q_{x,j}$}&&\ghost{X}\\
\lstick{$a_2$}&&\ghost{X}\\
\lstick{$a_1$}&&\ghost{X}
   \end{quantikz}\hspace{-2.cm}=
   \begin{quantikz}[row sep=0.5cm]
\lstick{$q_{x,1}$}&\gate[4,label
style={yshift=-0.3cm}]{{\rm U}\left(-\frac{\pi}{2}\right)^\dagger}&\ctrl{3}& \gate[4,label
style={yshift=-0.3cm}]{{\rm U}\left(-\frac{\pi}{2}\right)} & & \gate[5,label
style={yshift=-0.3cm}]{{\rm U}\left(-\frac{\pi}{2}\right)^\dagger}&\ctrl{3} &\gate[5,label
style={yshift=-0.3cm}]{{\rm U}\left(-\frac{\pi}{2}\right)} & &\ghost{X}\\
\wave&&&&&&&&&&&&\\
\lstick{$q_{x,j-1}$}&&\ctrl{1}& &&&\ctrl{1}&&&\ghost{X}\\
\lstick{$q_{x,j}$}&&\gate{{\rm RZ}\left(-\frac{\bar{u}\tau}{l}\right)}&&\gate{{\rm X}}&&\ctrl{1} &&\gate{{\rm X}}&\ghost{X}\\
\lstick{$a_2$}&&\gate{{\rm RY}\left(-\frac{2\tau}{\bar{\rho}l}\right)}&&&& \gate{{\rm RZ}\left(-\frac{2\tau}{\bar{\rho}l}\right)}& &&\ghost{X}\\
\lstick{$a_1$}&\gate{{\rm X}}&\ctrl{-1}&&&&\ctrl{-1}& \gate{{\rm X}}& &\ghost{X} 
   \end{quantikz}
   }
   \caption{Quantum circuit for the ${\rm \tilde{W}}_{x,j}$ operator acting on the $q_{x,1},\dots,q_{x,j},a_1,a_2$ register.}
   \label{fig:W_x_operator2}
\end{figure*}

Using first-order Lie-Trotter-Suzuki decomposition the time-evolution operator takes the form
\begin{align}\label{eq:evolution_operator2}
&{\rm \tilde{V}}(\tau)=\exp{\left(-i\tilde{H}\tau\right)}\nonumber\\
&\approx\exp{\left(-iH_x^{diag}\tau\right)}\exp{\left(-i\tilde{H}_x^{off}\tau\right)}\exp{\left(-i\tilde{H}_y^{off}\tau\right)}\nonumber\\
&\approx \prod_{j=1}^n {\rm \tilde{W}}_{x,j} \prod_{j=1}^n {\rm \tilde{W}}_{y,j},
\end{align}
where the operators ${\rm \tilde{W}}_{x,j}$ and ${\rm \tilde{W}}_{y,j}$ take the form
\begin{widetext}
\begin{align}\label{eq:evolution_operator22}
&{\rm \tilde{W}}_{x,j}= {\rm I}^{a_1}\otimes {\rm I}^{a_2}\otimes {\rm I}^{\otimes(n-j)}\otimes {\rm U}_j(-\pi/2)\exp{\left(i \frac{\bar{u}\tau}{2l}{\rm Z}\otimes \sigma_{11}^{\otimes (j-1)}\right)}{\rm U}_j(-\pi/2)^{\dagger}\otimes {\rm I}^{\otimes n}\nonumber\\
&\times   {\rm I}^{\otimes(n-j)}\otimes\left({\rm X}^{a_1}\otimes {\rm X}^j\otimes{\rm U}_{a_2}\left(-\frac{\pi}{2}\right)\right)\exp{\left( i\frac{\tau}{\bar{\rho}l}\sigma_{11}^{a_1}\otimes{\rm Z}^{a_2}\otimes \sigma_{11}^{\otimes j}\right)}\left({\rm U}_{a_2}\left(-\frac{\pi}{2}\right)^{\dagger}\otimes {\rm X}^j\otimes{\rm X}^{a_1}\right)\otimes{\rm I}^{\otimes n}\nonumber\\
&\times {\rm X}^{a_1}\exp{\left(i\frac{\tau}{\bar{\rho}l}\sigma_{11}^{a_1}\otimes {\rm Y^{a_2}}\right)}{\rm X^{a_1}}\otimes {\rm I}^{\otimes n} \otimes {\rm I}^{\otimes n},\nonumber\\
&{\rm \tilde{W}}_{y,j}=  {\rm I}^{\otimes n}\otimes{\rm I}^{\otimes(n-j)}\otimes\left({\rm X}^{a_2}\otimes {\rm X}_j\otimes{\rm U}_{a_1}\left(-\frac{\pi}{2}\right)\right)\exp{\left(i\frac{\tau}{\bar{\rho}l}{\rm Z}^{a_1}\otimes \sigma_{11}^{a_2}\otimes \sigma_{11}^{\otimes j}\right)}\left({\rm U}_{a_1}\left(-\frac{\pi}{2}\right)^{\dagger}\otimes {\rm X}^{j}\otimes {\rm X}^{a_2}\right)\nonumber\\
&\times {\rm I}^{\otimes n}\otimes {\rm X}^{a_2}\exp{\left(i\frac{\tau}{\bar{\rho}l}{\rm Y^{a_1}}\otimes\sigma_{11}^{a_2}\right)}{\rm X^{a_2}}\otimes {\rm I}^{\otimes n}.
\end{align}
\end{widetext}
In the above we recognize one of the factors as a controlled ${\rm RY}$ rotation with the control on the $a_1$ qubit and the target on $a_2$:
\begin{align}
    {\rm CRY}_{a_2}^{a_1}\left(-\frac{2\tau}{\bar{\rho}l} \right)=\exp{\left(i\frac{\tau}{\bar{\rho}l}\sigma_{11}^{a_1}\otimes {\rm Y^{a_2}}\right)},\nonumber
\end{align}

As before, we separately implement the $x$ and $y$ evolution terms ${\rm \tilde{W}}_{x,j}$ (see Fig.~\ref{fig:W_x_operator2}) and ${\rm \tilde{W}}_{y,j}$ (see Fig.~\ref{fig:W_y_operator2}). Note though the appearance of a factor of $2$ in the rotation angles. Comparing obtained circuits to the ones with only $D^\pm$ operator in the main text we see that no ${\rm MCRZZ}$ gates appear, which results in a more economical decomposition in terms of the elementary ${\rm CNOT}$ numbers. However, usage of three types of derivatives at once can disturb the symmetry of the solution.

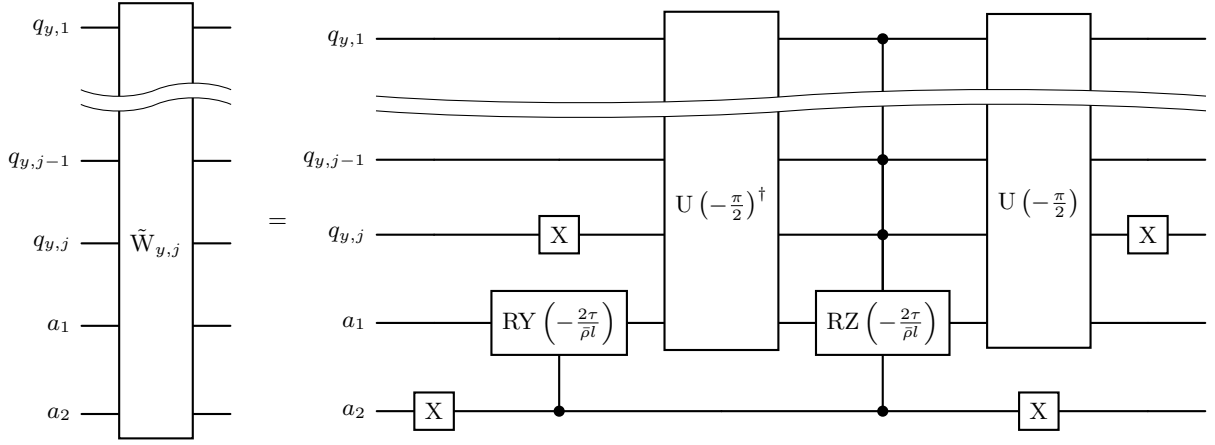
\begin{figure*}[ht]
    \centering
    \begin{quantikz}[row sep=0.6cm]
\lstick{$q_{y,1}$}&\gate[6,label
style={yshift=-0.3cm}]{{\rm \tilde{W}}_{y,j}}&\ghost{X}\\
\wave&&&&&&&\\
\lstick{$q_{y,j-1}$}&&\ghost{X}\\
\lstick{$q_{y,j}$}&&\ghost{X}\\
\lstick{$a_1$}&&\ghost{X}\\
\lstick{$a_2$}&&\ghost{X}
    \end{quantikz}\hspace{-2.cm}=
    \begin{quantikz}[row sep=0.5cm]
\lstick{$q_{y,1}$}&&&\gate[5,label
style={yshift=-0.3cm}]{{\rm U}\left(-\frac{\pi}{2}\right)^\dagger}&\ctrl{3} &\gate[5,label
style={yshift=-0.3cm}]{{\rm U}\left(-\frac{\pi}{2}\right)}&&\ghost{X}\\
\wave&&&&&&&&\\
\lstick{$q_{y,j-1}$}&&&& \ctrl{1}&&&\ghost{X}\\
\lstick{$q_{y,j}$}&&\gate{{\rm X}}&& \ctrl{1}&&\gate{{\rm X}}&\ghost{X}\\
\lstick{$a_1$}&&\gate{{\rm RY}\left(-\frac{2\tau}{\bar{\rho}l}\right)}& &\gate[1]{{\rm RZ}\left(-\frac{2\tau}{\bar{\rho}l}\right)} &&&\ghost{X}\\
\lstick{$a_2$}&\gate{{\rm X}}&\ctrl{-1}&& \ctrl{-2}&\gate{{\rm X}}&&\ghost{X}
    \end{quantikz}
    \caption{Quantum circuit for the ${\rm \tilde{W}}_{y,j}$ operator acting on the $q_{y,1},\dots,q_{y,j},a_1,a_2$ register.}
    \label{fig:W_y_operator2}
\end{figure*}

\section{Circuit derivation for the $D^{\pm}$ operator with periodic boundary conditions\label{appPeriodiccond}}
\setcounter{equation}{0}
\renewcommand{\theequation}{C\arabic{equation}}

The central finite difference operator with periodic boundary conditions can be encoded on $n$ qubits as follows :
\begin{align}\label{eq:diffopPBC}
D^{\pm}_P=\frac{1}{2l}\left(S^--S^+-\sigma_{01}^{\otimes n}+\sigma_{10}^{\otimes n}\right).
\end{align}
In the first order Trotter approximation the unitary evolution it generates is given by a product of two factors:
\begin{align}\label{eq:pbc1}
    &e^{D^\pm_{P} \gamma} = e^{-i (iD^\pm_P) \gamma}\approx \exp{\left[-i\left(\frac{i\gamma}{2l}(S^--S^+)\right)\right]}\nonumber\\
    &\times\exp{\left[-i\left(\frac{i\gamma}{2l}(-\sigma_{01}^{\otimes n}+\sigma_{10}^{\otimes n})\right)\right]},
\end{align}
where the first factor is given in the Eq.~\ref{eq:bc1} and the second one can be shown in a similar way to be:
\begin{align}\label{eq:pbc2}
    &\exp{\left[-i\left(\frac{i\gamma}{2l}(-\sigma_{01}^{\otimes n}+\sigma_{10}^{\otimes n})\right)\right]}\nonumber\\
    &=\left({\rm I}\otimes {\rm X}^{\otimes (n-1)}\right){\rm U}_n\left(-\frac{\pi}{2}\right){\rm MCRZ}_n^{1,\ldots n-1}\left(-\frac{\gamma}{l}\right)\nonumber\\
    &\times{\rm U}_n\left(-\frac{\pi}{2}\right)^{\dagger}\left({\rm I}\otimes {\rm X}^{\otimes (n-1)}\right).
\end{align}
Here we used the following identity, which can be easily verified (see Eq.~\ref{eq:sum_operator})
\begin{align}
&i(-\sigma_{01}^{\otimes n} +\sigma_{10}^{\otimes n})=\left({\rm I}\otimes {\rm X}^{\otimes (n-1)}\right){\rm U}_n\left(-\frac{\pi}{2}\right)\nonumber\\
&\times\left(-Z\otimes \sigma_{11}^{\otimes (n-1)}\right){\rm U}_n\left(-\frac{\pi}{2}\right)^{\dagger}\left({\rm I}\otimes {\rm X}^{\otimes (n-1)}\right).\nonumber\\
\end{align}

To write down the Euler equations with periodic boundary conditions we only need to add to the Dirichlet boundary conditions  case Eq.~\ref{eq:hamiltonian} an additional term representing ``corner" entries in the difference operator:
\begin{align}\label{eq:hamiltonianPBC}
   &H_P= -\frac{i}{2l} \left(\bar{u} {\rm I}^{a_1}\otimes {\rm I}^{a_2} + \frac{1}{\bar{\rho}}\sigma_{00}^{a_1}\otimes{\rm X}^{a_2} \right)\nonumber\\
   &\otimes (-\sigma_{01}^{\otimes n}+\sigma_{10}^{\otimes n})\otimes {\rm I}^{\otimes n}\nonumber\\
   &-\frac{i}{2l\bar{\rho}}{\rm X}^{a_1}\otimes\sigma_{00}^{a_2}\otimes {\rm I}^{\otimes n} \otimes (-\sigma_{01}^{\otimes n}+\sigma_{10}^{\otimes n}).
\end{align}
Making similar calculations as in the Appendix~\ref{appegen}, in the first order Trotter approximation we obtain the following evolution operator:
\begin{align}\label{evolutionopPBC}
{\rm V}(\tau)=e^{-iH\tau}\approx Q_y{\rm W}_{y,P} Q_x{\rm W}_{x,P},
\end{align}
where $Q_x$ and $Q_y$ are defined by Eqs.~\ref{eq:evolution_op2} and \ref{eq:evolution_operator4}, and the ${\rm W}_{x,P}$, ${\rm W}_{y,P}$ operators are given by
\begin{widetext}
\begin{align}
&{\rm W}_{x,P}= {\rm I}^{a_1}\otimes {\rm I}^{a_2}\otimes  \left({\rm I}\otimes {\rm X}^{\otimes (n-1)}\right) {\rm U}_n\left(-\frac{\pi}{2}\right) {\rm MCRZ}_n^{1,\ldots, n-1}\left(\frac{\bar{u}\tau}{l}\right) 
{\rm U}_n\left(-\frac{\pi}{2}\right)^{\dagger} \left({\rm I}\otimes {\rm X}^{\otimes (n-1)}\right)
\nonumber\\
&\times  
\left({\rm I\otimes {\rm X}^{\otimes (n-1)}}\right) 
\left({\rm X}^{a_1}\otimes {\rm H}^{a_2}
\otimes{\rm U}_{n}\left(-\frac{\pi}{2}\right)
\right)
{\rm MCRZZ}_{n,a_2}^{1,\ldots, n-1,a_1}\left(\frac{\tau}{\bar{\rho}l}\right) \left({\rm X}^{a_1}\otimes {\rm H}^{a_2}\otimes{\rm U}_{n}\left(-\frac{\pi}{2}\right)^{\dagger}\right)\left({\rm I\otimes {\rm X}^{\otimes (n-1)}}\right)
\otimes {\rm I}^{\otimes n},\nonumber\\
&{\rm W}_{y,P}=  {\rm I}^{\otimes n}\otimes\left({\rm I}\otimes {\rm X}^{\otimes (n-1)}\right)\left({\rm H}^{a_1}\otimes {\rm X}^{a_2}\otimes {\rm U}_{n}\left(-\frac{\pi}{2}\right)\right) {\rm MCRZZ}_{n,a_1}^{1,\ldots, n-1,a_2}\left(\frac{\tau}{\bar{\rho}l}\right) \left({\rm H}^{a_1}\otimes {\rm X}^{a_2}\otimes{\rm U}_{n}\left(-\frac{\pi}{2}\right)^{\dagger}\right)\left({\rm I}\otimes {\rm X}^{\otimes(n-1)}\right).\nonumber
\end{align}
\end{widetext}
Note the changed sign in the multi-controlled rotations compared to circuit in the main text.

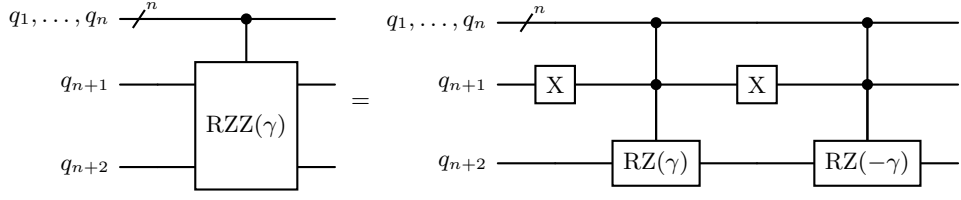
\begin{figure*}[!t]
    \centering
    \begin{quantikz}
\lstick{$q_1,\dots,q_n$}&\qwbundle{n}&\ctrl{1}&\\
\lstick{$q_{n+1}$}&&\gate[2]{{\rm RZZ}(\gamma)}&\\
\lstick{$q_{n+2}$}&&&
\end{quantikz}
=
\begin{quantikz}
\lstick{$q_1,\dots,q_n$}&\qwbundle{n}&\ctrl{2}&&\ctrl{2}&\\
\lstick{$q_{n+1}$}&\gate{{\rm X}}&\ctrl{1}&\gate{{\rm X}}&\ctrl{1}&\\
\lstick{$q_{n+2}$}&&\gate{{\rm RZ}(\gamma)}&&\gate{{\rm RZ}(-\gamma)}&
\end{quantikz}
    \caption{A decomposition of ${\rm MCRZZ}^{1,\dots,n}_{n+1,n+2}$ gate into two ${\rm MCRZ}^{1,\dots,n+1}_{n+2}$ gates.}
    \label{fig:mcrzz_decompozition}
\end{figure*}

\section{Proof of Lemma~\ref{lemma1} \label{appCommutators}}
\setcounter{equation}{0}
\renewcommand{\theequation}{D\arabic{equation}}

Consider the Hamiltonian Eq.~\ref{eq:hamiltonian} in the form given in Eqs.~\ref{eq:hamiltonian3}-\ref{eq:hamiltonian6}. From \cite[Proposition 9]{Childs2021} and the triangle inequality the additive Trotter error between $\exp(-iH\tau)$ and $V(\tau)$ can be bounded by
 \begin{align}\label{eq:upper_error_bound}
     &\sum_{j=1}^n \sum_{j'=j+1}^n\left\lVert\left[H^{diag}_{x,j},H^{diag}_{x,j'}\right]\right\rVert
     +\sum_{j=1}^n \sum_{j'=j+1}^n\left\lVert\left[H^{off}_{x,j},H^{off}_{x,j'}\right]\right\rVert\nonumber\\
     &+\sum_{j=1}^n \sum_{j'=j+1}^n\left\lVert\left[H^{off}_{y,j},H^{off}_{y,j'}\right]\right\rVert
     + \sum_{j=1}^n \sum_{j'=1}^n\left\lVert\left[H^{diag}_{x,j},H^{off}_{x,j'}\right]\right\rVert\nonumber\\
     &+\sum_{j=1}^n \sum_{j'=1}^n\left\lVert\left[H^{diag}_{x,j},H^{off}_{y,j'}\right]\right\rVert
     + \sum_{j=1}^n \sum_{j'=1}^n\left\lVert\left[H^{off}_{x,j},H^{off}_{y,j'}\right]\right\rVert,
 \end{align}
where the terms with subindex $j$ or $j'$ correspond to elements under the sum in Eq.~\ref{eq:hamiltonian4}-\ref{eq:hamiltonian6}.

In Refs.~\cite{toyota2024pde,Hu_2024} it was shown that two terms Eq.~\ref{eq:sum_operator} with indices $j'>j$ do not commute with each other only when $j=1$; in this case the norm of their commutator is $1$.
From this we can instantly give estimates for first three sums in Eq.~\ref{eq:upper_error_bound}:
\begin{align}\label{eq:upper_error123}
     &\sum_{j=1}^n \sum_{j'=j+1}^n\left\lVert\left[H^{diag}_{x,j},H^{diag}_{x,j'}\right]\right\rVert = \left(\frac{\bar{u}}{2l}\right)^2(n-1),\nonumber\\
     &\sum_{j=1}^n \sum_{j'=j+1}^n\left\lVert \left[H^{off}_{x,j},H^{off}_{x,j'}\right]\right\rVert = \left(\frac{1}{2l\bar{\rho}}\right)^2(n-1),\nonumber\\
     &\sum_{j=1}^n \sum_{j'=j+1}^n \left\lVert\left[H^{off}_{y,j},H^{off}_{y,j'}\right]\right\rVert = \left(\frac{1}{2l\bar{\rho}}\right)^2(n-1).
\end{align}
Let us calculate the fourth term in Eq.~\ref{eq:upper_error_bound}. Using the explicit form of the Hamiltonian $H^{diag}_{x,j}$ and $H^{off}_{x,j'}$ the commutator between them takes the form
 \begin{widetext}
 \begin{align}\label{eq:upper_error4}
 &\left[H^{diag}_{x,j},H^{off}_{x,j'}\right] = -\frac{\bar{u}}{\bar{\rho}}\left(\frac{1}{2l}\right)^2\sigma_{00}^{a_1}\otimes X^{a_2}\nonumber\\
 &\otimes\left[{\rm I}^{\otimes (n-j)}\otimes (\sigma_{01}\otimes\sigma_{10}^{\otimes(j-1)}-\sigma_{10}\otimes\sigma_{01}^{\otimes(j-1)}),{\rm I}^{\otimes (n-j')}\otimes (\sigma_{01}\otimes\sigma_{10}^{\otimes(j'-1)}-\sigma_{10}\otimes\sigma_{01}^{\otimes(j'-1)})\right]\otimes {\rm I}^{\otimes n}.
 \end{align}
 \end{widetext}
 Upon a closer inspection we recognize that this commutator has precisely the same structure as the ones considered before, therefore its norm can be evaluated to be
 \begin{align}\label{eq:upper_error5}
     \sum_{j=1}^n \sum_{j'=1}^n\left\lVert\left[H^{diag}_{x,j},H^{off}_{x,j'}\right]\right\rVert = \frac{\bar{u}}{\bar{\rho}}\frac{2}{\left(2l\right)^2}(n-1).
 \end{align}
 
It is easy to verify that the commutator between the $H^{diag}_{x,j}$ and $H^{off}_{y,j}$ operators is equal to zero. This is because they define the evolution on disjoint sets of qubits.
 Therefore, we obtain
 \begin{align}\label{eq:upper_error6}
     \sum_{j=1}^n \sum_{j'=1}^n\left\lVert\left[H^{diag}_{x,j},H^{off}_{y,j'}\right]\right\rVert = 0.
 \end{align}

\begin{figure*}[!t]
    \centering
    \includegraphics[width=\textwidth]{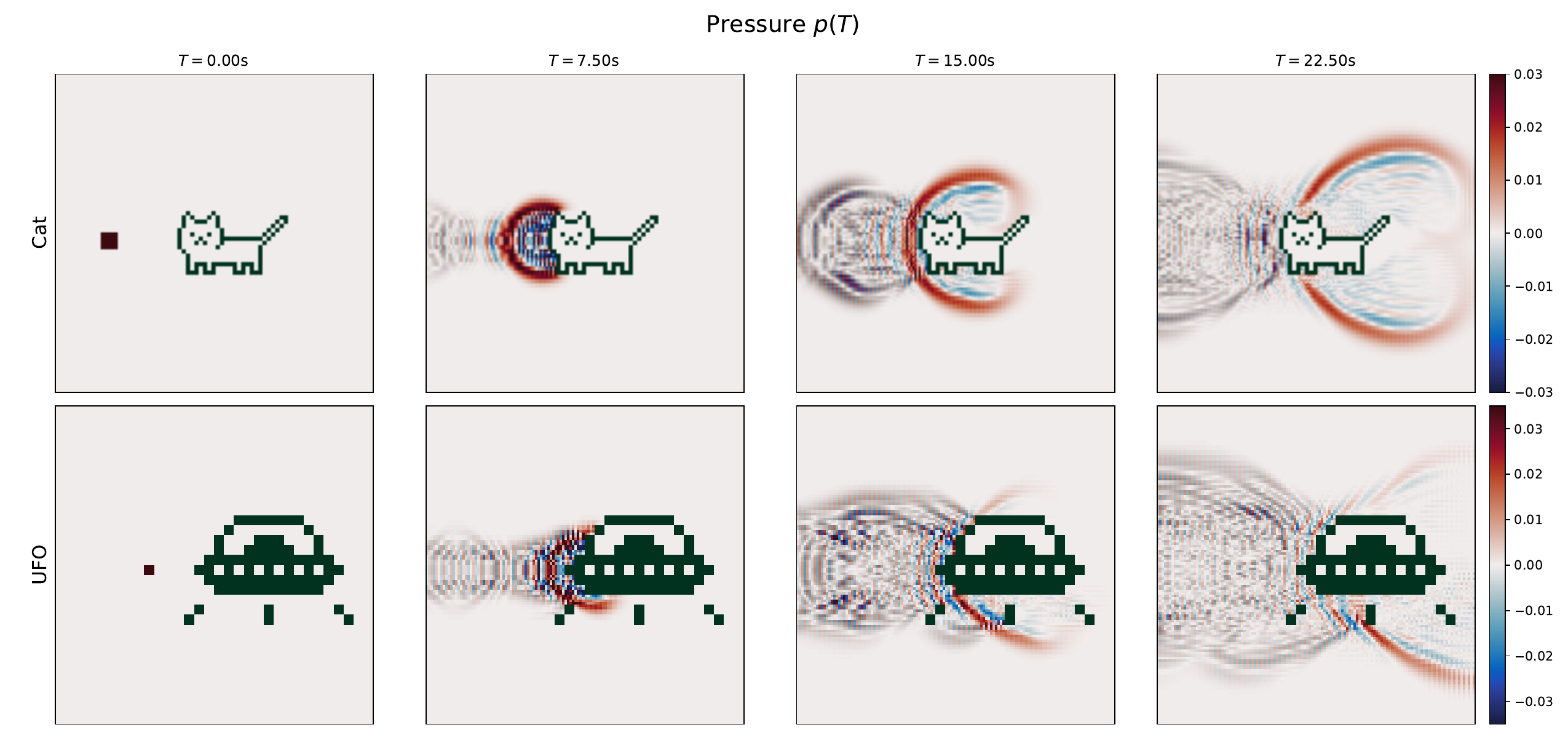}
    \caption{Quantum simulation of an acoustic wave propagating on top of a constant background flow and scattering from a complex shaped obstacle, obtained by executing quantum circuits derived using our construction. Here we plot the pressure (difference) at four distinct time points for the cat- \textbf{(top row)} and UFO-shaped obstacles \textbf{(bottom row)}. We used $n=8$ qubits, mean flow velocity $\bar{u}=2$ and a time step of $\tau=0.05$.}
    \label{fig:cat_ufo}
\end{figure*}
 
Finally, let us compute the commutator between the $H^{off}_{x,j}$ and $H^{off}_{y,j}$ operators:
  \begin{widetext}
 \begin{align}\label{eq:upper_error7}
 &\left[H^{off}_{x,j},H^{off}_{y,j'}\right] = -\frac{1}{(2l\bar{\rho})^2}\left[\sigma_{00}^{a_1}\otimes X^{a_2},X^{a_1}\otimes\sigma_{00}^{a_2}\right]\nonumber\\
 &\otimes{\rm I}^{\otimes (n-j)}\otimes (\sigma_{01}\otimes\sigma_{10}^{\otimes(j-1)}-\sigma_{10}\otimes\sigma_{01}^{\otimes(j-1)})\otimes{\rm I}^{\otimes (n-j')}\otimes (\sigma_{01}\otimes\sigma_{10}^{\otimes(j'-1)}-\sigma_{10}\otimes\sigma_{01}^{\otimes(j'-1)})\nonumber\\
 &=-\frac{1}{(2l\bar{\rho})^2}\left(\sigma_{01}^{a_1}\otimes \sigma_{10}^{a_2}-\sigma_{10}^{a_1}\otimes\sigma_{01}^{a_2}\right)\nonumber\\
 &\otimes{\rm I}^{\otimes (n-j)}\otimes (\sigma_{01}\otimes\sigma_{10}^{\otimes(j-1)}-\sigma_{10}\otimes\sigma_{01}^{\otimes(j-1)})\otimes{\rm I}^{\otimes (n-j')}\otimes (\sigma_{01}\otimes\sigma_{10}^{\otimes(j'-1)}-\sigma_{10}\otimes\sigma_{01}^{\otimes(j'-1)}).
 \end{align}
 \end{widetext}
 The spectral norm of this commutator is equal to $1/(2l\bar{\rho})^2$, therefore we have:
\begin{align}\label{eq:upper_error8}
     &\sum_{j=1}^n \sum_{j'=1}^n \left\lVert\left[H^{off}_{x,j},H^{off}_{y,j'}\right]\right\rVert \leq  \frac{1}{(2l\bar{\rho})^2}\sum_{j=1} \sum_{j'= 1} 1\nonumber\\
     &=\frac{1}{(2l\bar{\rho})^2}n^2.
\end{align}

Summing everything up we obtain Eq.~\ref{eq:upper_error}, which finishes the proof.
\qed

\section{Improved obstacle circuit implementation on a QPU}\label{appObstacle}
\setcounter{equation}{0}
\renewcommand{\theequation}{E\arabic{equation}}

In Sec.~\ref{object_ins_env} we described the general algorithm for implementing arbitrary-shaped obstacles in the flow.
For executing simulation on near- to mid-term QPUs, however, even constant factor improvements can make a crucial difference. Here we describe several such optimizations, which reduce the circuit depth and other sources of error (\emph{e.g.}~the Trotter error).

First, we show that obstacles can be implemented without increasing the overall Trotter error, which is not immediately apparent from the form of the generator of evolution with obstacles in Eq.~\ref{eq:dpmsubs}. To this end
recall that a rectangular binary cell obstacle with Dirichlet boundary conditions was implemented in Sec.~\ref{object_ins_env} by canceling the action of the finite difference operator on its boundaries. The difference operator $D^\pm$ was given in Eq.~\ref{eq:bc1} as a product of factors which we here denote by $W_j(\gamma)$:
\begin{equation}\label{eq:W_j_general}
    W_j(\gamma) = {\rm U}_j(-\pi/2) {\rm MCRZ}^{1,\dots, {j-1}}_{j}\left(\frac{\gamma}{l}\right) {\rm U}_j(-\pi/2)^\dagger.
\end{equation}

As mentioned in Sec.~\ref{object_ins_env} each $W_j(\gamma)$ implements $2^{n-j}$ disjoint sets of $\pm1$ elements on the super- and sub-diagonals in the explicit matrix representation of $D^\pm$. Thus, cancellation of these $\pm1$ pairs by the obstacle terms only affects the corresponding $W_j$ operator. Furthermore, we have seen that the obstacle cancellation factor $e^{\sigma_{\hat{n}}\gamma}$ in Eq.~\ref{eq:sigma_exp} is also equal to $W_{j}$ with $j=\hat{n}$, up to the controls on both registers which depend on the location of the obstacle, \emph{i.e.}~that $\sigma_{\hat{n}}$ is the generator of $W_{\hat{n}}$.

Note now that if arbitrary generators $A$ and $B$ commute, then so do their controlled versions ${\rm I}\otimes A$, $|1\rangle\langle1|\otimes B$ and $|0\rangle\langle0|\otimes B$. This is also true of their multi-controlled versions, as easily shown by the mathematical induction. Taking advantage of the above observation, we can group all the generators of the $D^\pm$ evolution factors $W_j$ in Eq.~\ref{eq:bc1line} with the obstacle terms sharing the same value of $\hat{n}=j$ (a general obstacle can be composed of a multitude of binary cells of different sizes). Since for each $j$ we now have the generator of $W_j$ and (possibly) its multi-controlled version from the obstacle, which commute, the Trotter implementation of their sum is \emph{exact}. The only Trotter error comes from non-commuting terms for different $j$ (see App.~\ref{appCommutators}) and does not increase, since putting some of the matrix elements to zero does not increase its spectral norm. This is to be contrasted with direct Trotterization of the evolution generated by Eq.~\ref{eq:dpmsubs}, where there would be additional error from the un-rearranged obstacle terms.

The above technical manipulation has a very clear physical interpretation: the exact implementation of the obstacle in the Trotter procedure ensures that the boundaries of the obstacle remain \textit{impenetrable}, which would not hold in the naive Trotterization of Eq.~\ref{eq:dpmsubs}.

In addition to the Trotter error improvement, a number of circuit simplifications are possible. First, note that in the construction of the $D^\pm$ operator with obstacles there is no need to implement the hermitian conjugate of $e^{\sigma_{\hat{n}}\gamma}$ as an inverse circuit (see Fig.~\ref{fig:D_operator_obstacle}). Instead it suffices to change the sign in the ${\rm RZ}$ rotation in Eq.~\ref{eq:W_j_general} or Eq.~\ref{eq:sigma_exp}. 

Next, observe that the implementation of the obstacle circuit in Fig.~\ref{fig:D_operator_obstacle} contains multi-controlled versions of the multi-target operator $e^{\sigma_{\hat{n}}\gamma} \equiv W_{\hat{n}}$. The naive implementation of it would put the multi-controls on all of the factors in the definition in Eq.~\ref{eq:W_j_general}, resulting in a very complex circuit (after transpilation to a local gate basis). Instead we note that the effect of putting the controls on the $U_j(-\pi/2)$ operator and its inverse is redundant, and exactly the same the same action is obtained by only putting the controls on the ${\rm MCRZ}$ gate. Since this is in itself a multi-controlled gate, we obtain an MCRZ gate with a larger control set. The total amount of controls may not exceed the register size and remains $\mathcal{O}(n)$, and therefore the total circuit complexity remains quadratic (see the proof of Lemma~\ref{lemma2}).

Recalling that in the re-arrangement minimizing the Trotter error all terms of the same size $j$ are grouped together, we observe that the ``inner" basis change operators $U_j(\pi/2)$ coming from the original generator of $D^\pm$ and from the obstacle cancel exactly (since the latter can also be made un-controlled, as discussed above). There is thus a single pair of $U_j(\pi/2)$, $U_j(\pi/2)^\dagger$ per each $j$ group left, further simplyfing the circuit.

Finally, we note that for a more complex obstacle shape which decomposes into multiple binary cells of different sizes, potentially adjacent to each other, there is no need to include cancellation terms for the boundaries they share.

\section{Additional simulations with complex obstacles \label{appExperiments}}

To illustrate the applicability of our circuit construction to general obstacle types we supplement the experiments in Section~\ref{experiment} with additional demonstration, with more complex -- if not entirely realistic -- shapes, the outlines of which form a cat and a UFO. The results of executing the circuits for the acoustic wave propagation in a background flow obtained using our algorithm are shown in  Fig.~\ref{fig:cat_ufo}. Here we used $c=\bar{\rho}=1$, $\bar{u}=2$, $l=0.5$, $\tau=0.05$ and employed $n=8$ qubits.

\begin{figure}
    \centering
    \includegraphics[width=\textwidth]{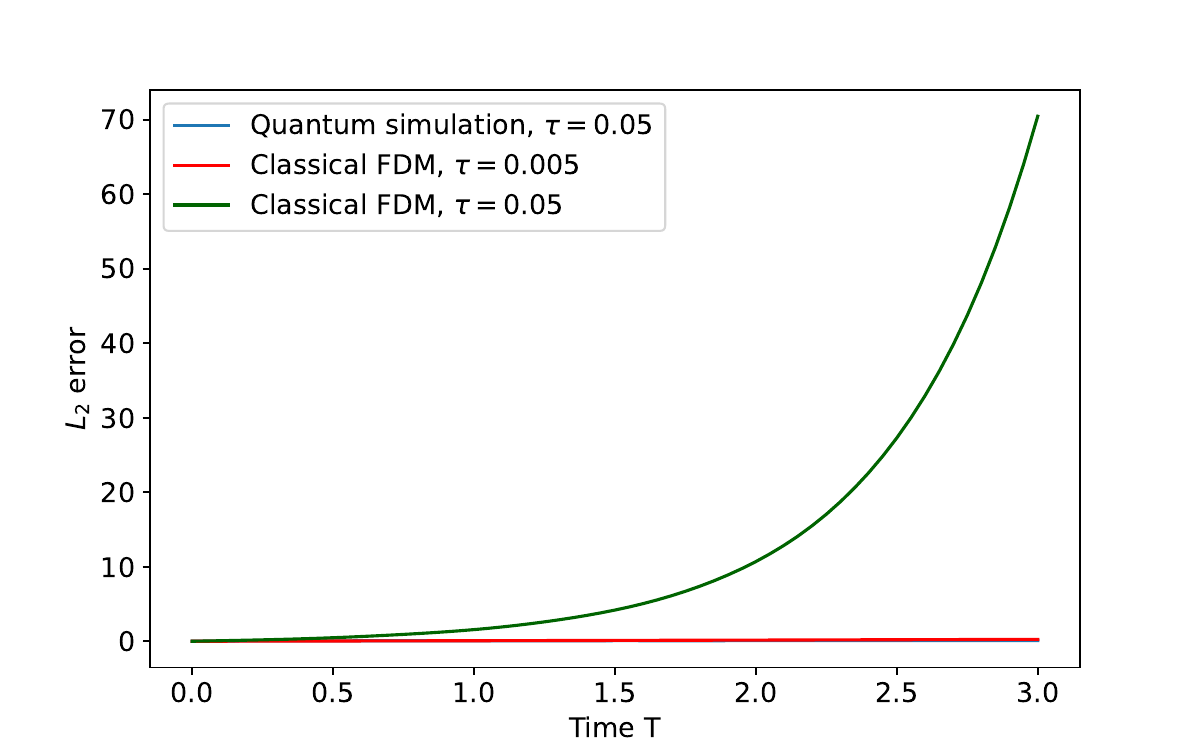}
    \caption{The $L_2$-error of the quantum and classical FDM solutions for $p(T)$ in Fig.\ref{fig:euler_ideal_comparison} with respect to the exact matrix exponential solution. When using the same time increment of $\tau=0.05$ as the quantum solution the classical method becomes numerically unstable. To ensure a comparable (but still larger) error the classical FDM required $\tau=0.005$ (see Fig.\ref{fig:l2_error})}
    \label{fig:l2_error_unstable}
\end{figure}

\section{Extending quantum simulations to non-conservative regimes\label{appNonlinear}}
\setcounter{equation}{0}
\renewcommand{\theequation}{G\arabic{equation}}

Here we give a brief overview of how simulation of PDEs, in particular of LEE Eq.~\ref{eq:pde_euler_full}, can be extended to the case when then underlying numerical simulation can not be written as the Schr\"odinger equation directly. For the LEE this happens when $c\neq 1/\bar{\rho}$ or different boundary conditions, \emph{e.g.}~von Neumann, are used.

Let $A$ be the matrix, which represents the numerical scheme of the simulation on a grid with already substituted differential operators and boundary conditions. We can split it in its real and imaginary parts by:
\begin{equation}\label{eq^hermitizationofA}
    A_1 = \frac{A + A^\dagger}{2}, \quad\quad\quad A_2 = \frac{A - A^\dagger}{2i},
\end{equation}
where both matrices are Hermitian ($A_{1,2}^\dagger = A_{1,2}$). Using the first order Trotter formula we have:
\begin{equation}
    e^{A\tau} = e^{A_1\tau + iA_2\tau} \approx e^{A_1\tau}  e^{iA_2\tau}.
\end{equation}
Thus, simulation of the PDE on a quantum computer can be reduced to implementing alternating steps of imaginary and real time evolutions. Efficient implementation of such quantum circuits is an active area of research, and multiple methods already exist.
In particular,  Refs.~\cite{Joseph_2020,Novikau_2025} extend the PDE with additional phase space via Koopman-von Neumann approach, so that solution can be stored and implemented on a quantum computer. A similar idea of extending the solution space with additional dimension to implement $e^{A_1\tau}$ is the basis of the ``Schr\"odingerization" method {\cite{satohiroukinaoki2024,jinliuma2023,PhysRevLett.133.230602,jin2025schrodingerizationbasedquantumalgorithms,jin2023quantumsimulationquantumdynamics,jin2024quantumsimulationfokkerplanckequation}. In \cite{jin2022quantumsimulationpartialdifferential,Hu_2024} this technique is applied to non-conservative heat and advection equations. More general imaginary time evolutions methods can also be used \cite{leamer2024quantumdynamicalemulationimaginary,Motta_2019,McArdle_2019}, in particular methods based on linear combination of unitaries (LCU) are the most common \cite{leadbeater2023nonunitarytrottercircuitsimaginary,An_2023}. In \cite{guseynov2024explicitgateconstructionblockencoding} Schr\"odingerization and LCU methods are combined and generalized, using block-encoding for Hamiltonians. 

Importantly, in all of these techniques the Hilbert space is extended, and unitary operators such as $e^{iA_1\tau}$ (notice the appearance of the imaginary number $i$) are used as the building blocks in the implementation of the required non-unitary $e^{A_1\tau}$. Typically the unitaries become target operators controlled by variables from the extended space using quantum matrix arithmetics \cite{Gily_n_2019}. These controlled versions are still implementable with polynomial resources, therefore the quantum advantage in simulating the exponentially large grid outlined in Section~\ref{comp_analais} remains. Thus, our results on efficient implementation of the unitary operators are also directly relevant to the more general non-unitary case.

For LEE the matrix $A_2$ is structurally the same as in Eq.~\ref{eq:euler_evolution}, but with off-diagonal coefficients equal to $(1/\bar{\rho} + \bar{\rho}c^2) / 2i$. $A_1$ is also structurally similar, except with $(1/\bar{\rho} - \bar{\rho}c^2) / 2$ off-diagonal coefficients and a vanishing diagonal part. Consequently, $e^{iA_1\tau}$  can be implemented following the prescription described in the main text with removing the gates corresponding to the now missing diagonal part. Note that in the conservative regime ($c=1/\bar{\rho}$) we have $A_1=0$ and $e^{iA_2\tau}=e^{-iH\tau}$ for $H$ defined in Eq.~\ref{eq:hamiltonian}. We defer a more explicit derivation to the Appendix~\ref{appDirichletNeumanncond} below, as it is exactly the same as the one for more complex boundary conditions discussed there.

\section{Circuit derivation for the $D^{\pm}$ operator with Dirichlet and Neumann boundary condition}
\label{appDirichletNeumanncond}
\setcounter{equation}{0}

\renewcommand{\theequation}{H\arabic{equation}}

\begin{figure*}[ht]
    \centering
\begin{adjustbox}{width=\textwidth}
    \begin{quantikz}[row sep=0.5cm]
\lstick{$q_{1}$}&\gate{{\rm Ph}\left(\frac{\gamma}{4l}\right)}&\gate{{\rm H}}&\gate[1]{{\rm RZ}\left(-\frac{\gamma}{2l}\right)}&\gate{{\rm H}}&\ghost{X}\\
\lstick{$q_{2}$}&&&\ctrl{-1}&&\ghost{X}\\
\lstick{$q_{3}$}&&& \ctrl{-2}&&\ghost{X}\\
\wave&&&&&\\
\lstick{$q_{n-1}$}&&& \ctrl{-4}&&\ghost{X}\\
\lstick{$q_{n}$}&&& \ctrl{-5} &&\ghost{X}
    \end{quantikz}
\begin{quantikz}[row sep=0.5cm]
\lstick{$q_{1}$}  & \gate[6,label style={yshift=+0.3cm}]{e^{D^{\pm}}} & \gate{{\rm P}\left(-\frac{\pi}{2}\right)^\dagger} & \gate{{\rm H}}  & \gate{{\rm RZ}\left(-\frac{\gamma}{2l}\right)}  & \gate{{\rm H}} & \gate{{\rm P}\left(-\frac{\pi}{2}\right)} &  \ghost{X} \\
\lstick{$q_{2}$}   & \ghost{X} &  &   & \ctrl{-1}  &  &  &  \ghost{X} \\
\lstick{$q_{3}$} & \ghost{X} &  &   & \ctrl{-2}  &  &  &  \ghost{X} \\
\wave&&&&&&&&\\
\lstick{$q_{n-1}$} &  &  &  & \ctrl{-4}  &  &  & \ghost{X} \\
\lstick{$q_{n}$}   &  &  &   & \ctrl{-5} &   &  & \ghost{X}
\end{quantikz}
    \caption{\textbf{(Left)} The quantum circuit implementing the unitary operator $V(\gamma)=e^{i \left(A_1 + {\rm I}/(4l)\right) \gamma}$. The position of the global phase gate ${\rm Ph}$ can be arbitrary. \textbf{(Right)} The gate implementation of the operator $e^{i A_2 \gamma}$. 
    }
    \label{fig:diffoperatorsDNconds}
\end{adjustbox}
\end{figure*}

\begin{figure*}[!t]
    \centering
    \begin{quantikz}
\lstick{$q_1,\dots,q_n$}&\qwbundle{n}&\gate[1]{{\rm e^{iA_2\gamma}}}&\gate[1]{{\rm V}(\gamma)}&\gate[1]{{\rm V}(2\gamma)}&&\gate[1]{{\rm V}(2^{n_{\mathrm{lcu}}-1}\gamma)}& \gate[1]{{\rm V(-2^{n_{\mathrm{lcu}-1}}\gamma)}}&\\
\lstick{$a_1$}&&\gate[4]{{\rm PREP}}&\ctrl{-1} &&&&\gate[4]{{\rm PREP}^{\dagger}}&\\
\lstick{$a_2$}&&&&\ctrl{-2}&&&&\\
\wave&&&&&&&&\\
\lstick{$a_{n_{\mathrm{lcu}}}$}&&&&&&\ctrl{-4}&&
\end{quantikz}
     \caption{The quantum circuit implementing the operator in Eq.~\eqref{eq:LCU_representation3}. The PREP operator prepares the ancilla register in the state given by Eq.~\eqref{eq:LCU_prepancila}. The quantum circuits for the operators ${\rm V}(\gamma)$ and $e^{iA_2\gamma}$ are depicted in the left and right panels of Fig.~\ref{fig:diffoperatorsDNconds}, respectively.}
    \label{fig:LCU_decompozition}
\end{figure*}

In the main text we focused on examples using homogeneous Dirichlet boundary conditions (which, we note, typically require a ghost point implementation). Here we show how obstacles using Neumann or mixed boundary conditions can be implemented as well, using Linear Combination of Unitaries (LCU). To reduce the redundancy of presentation we describe in detail the construction for the more complex mixed boundary conditions in the case of central finite-difference operator approximating 
the first derivative on the interval $[0,L]$ (see Section~\ref{sec:differenceops}): Dirichlet boundary conditions at the left endpoint and a Neumann boundary condition 
at the right endpoint, that is $u(0)=0$ and $u'(L)=0$.

The matrix representation of the central difference approximation of the derivative operator with the above boundary conditions on a uniform grid with spacing $l$ is given by:

\begin{equation}\label{eq:DpmdirichletNeuman}
D^{\pm}_{DN}
=
\frac{1}{2l}
\begin{pmatrix}
0 & 1 & 0 & 0 & \dots & 0 \\
-1 & 0 & 1 & 0 & \dots & 0 \\
0 & -1 & 0 & 1 & \dots & 0 \\
\vdots & & \ddots & \ddots & \ddots & \vdots \\
0 & \dots & 0 & -1 & 0 & 1 \\
0 & \dots & 0 & 0 & 0 & 0
\end{pmatrix},
\end{equation}
where due to the Neumann boundary condition $u'(L)=0$ the last row consists entirely of zeros, reflecting the absence of flux through the right boundary.

As $D^{\pm}_{DN}$ is no longer antisymmetric, it cannot be hermitized simply by multiplying by the imaginary unit, as was the case for the $D^{\pm}$ operator in Eq.~\ref{eq:bc1line}. We therefore use the LCU construction used by Ref.~\cite{An_2023} for simulating non-unitary dynamics.
As described in the main text in Eq.~\ref{eq^hermitizationofA} the operator can be decomposed into its Hermitian and anti-Hermitian parts as 
$D^{\pm}_{DN} = A_1 + iA_2$, where
\begin{equation}
A_1=\frac{D^{\pm}_{DN}+\left(D^{\pm}_{DN}\right)^\dagger}{2},
\qquad
A_2=\frac{D^{\pm}_{DN}-\left(D^{\pm}_{DN}\right)^\dagger}{2i}.
\label{eq:hermeitianoperatorsDNcond}
\end{equation}
From this we have
\begin{align}
A_1&=\frac{1}{4l}\sigma_{11}^{\otimes (n-1)}\otimes(\sigma_{01} + \sigma_{10}),\\
A_2&=\frac{1}{i}D^\pm - \frac{1}{4li}\sigma_{11}^{\otimes (n-1)}\otimes(\sigma_{01} - \sigma_{10}).
\end{align}
To simulate the non-unitary evolution generated by the differentiation operator we use Theorem 1 of Ref.~\cite{An_2023}. Since the theorem requires the Hermitian matrix $A_1$ to also be positive semidefinite, we shift it by its smallest eigenvalue multiplied by an identity. The resulting operator $A_1+{\rm I}/(4l)$ is positive semidefinite operator and the corresponding non-unitary evolution admits an integral representation over a real parameter $k$
\begin{equation}
e^{D^{\pm}_{DN}\gamma} = e^{-\frac{\gamma}{4l}{\rm I}}\int_{-\infty}^{+\infty} \frac{1}{\pi(1+k^2)}e^{i\left(k\left(A_1+\frac{1}{4l}{\rm I}\right)+A_2\right)\gamma} {\rm d}k,
\label{eq:hermeitianoperatorsDNcond2}
\end{equation}
where to compensate for the spectral shift introduced in $A_1$, we multiplied the integral representation by the global factor $e^{-\gamma/(4l){\rm I}}$.  Instead of directly implementing the exponential of the operator $D^{\pm}_{DN} = A_1 + iA_2$, the evolution is rewritten as a continuous superposition of unitaries generated by the Hermitian operators $A_2 + k(A_1+{\rm I}/(4l))$. For each fixed real parameter $k$, the operator $e^{i\left(k(A_1+{\rm I}/(4l)) + A_2\right)\gamma}$ is unitary, and the non-unitary dynamics emerges after integration over $k$ with the weight $1/\left(\pi(1+k^2)\right)$.

To implement Eq.\ref{eq:hermeitianoperatorsDNcond2} in practice, the $k$-integral is discretized and a Trotter approximation is used:
\begin{align}
&e^{D^{\pm}_{DN}\gamma} \approx \left(\sum_{k=0}^{M-1}
c_k e^{i \left(k-\frac{M}{2}\right) \left(A_1 + \frac{1}{4l}{\rm I}\right) \gamma}\right) \ e^{i A_2\gamma}\nonumber\\
&=\left(e^{i\left(A_1 + \frac{1}{4l}{\rm I}\right) \gamma)}\right)^{-M/2}\left(\sum_{k=0}^{M-1}
c_k e^{i k \left(A_1 + \frac{1}{4l}{\rm I}\right) \gamma}\right) \ e^{i A_2\gamma },
\label{eq:LCU_representation}
\end{align}
where $\{k\}_0^{M-1}$ are the discrete integration points, whose total number $M = 2^{n_{\mathrm{lcu}}}$ is determined by the number of ancilla qubits $n_{\mathrm{lcu}}$, and the LCU coefficients are defined as $c_k = \frac{1}{\pi(1 + (k-M/2)^2)}$.
The evolution is thus decomposed into a unitary generated by $A_2$, acting on the system register, and an LCU decomposition of unitaries generated by $A_1+{\rm I}/(4l)$. The LCU part is implemented via 
ancilla-controlled operations weighted by the coefficients $c_k$, following the circuit construction of Ref.~\cite{satohiroukinaoki2024}.

Specifically, we first prepare the ancilla register in a superposition state, whose amplitudes encode the LCU coefficients $c_k$, using the standard $\mathrm{PREP}$ operator \cite{satohiroukinaoki2024,Hu_2024,Jin2024heat}:
\begin{equation}
\mathrm{PREP}\,|{\bf 0}\rangle  = \sum_{k=0}^{M-1} \sqrt{c_k}\vert k\rangle,
\label{eq:LCU_prepancila}
\end{equation}
where $\vert k\rangle$ are the computational basis states of the ancilla register.
At the end of the evolution the inverse operator $\mathrm{PREP}^\dagger$ is applied. When the ancilla register is in the $|{\bf 0}\rangle$ state the system register $q$ contains the state corresponding to the non-unitary evolution. The sum of unitaries in equation Eq.~\eqref{eq:LCU_representation} is then rewritten using ancilla-controlled unitaries as:
\begin{align}
&\sum_{k=0}^{M-1}
c_k e^{i k \left(A_1 + \frac{1}{4l}{\rm I}\right) \gamma}=\sum_{k=0}^{M-1} c_k \vert k\rangle\langle k\vert \otimes e^{i \gamma k \left(A_1 + \tfrac{1}{4l} I\right)}\nonumber\\
&=\mathrm{PREP}^{\dagger}\prod_{j=0}^{n_{\mathrm{lcu}}-1}
\left(|0\rangle\langle0|_j \otimes I+|1\rangle\langle1|_j \otimes V^{2^j}(\gamma)\right)\mathrm{PREP},
\label{eq:LCU_decompositionsch}
\end{align}
where $V(\gamma)=\exp\!\left(i \gamma \left(A_1 + \tfrac{1}{4l} I\right)\right)$.
Now we can write the total form of the operator Eq.~\ref{eq:LCU_representation}
\begin{align}
&e^{D^{\pm}_{DN}\gamma}\approx V^{-2^{n_{\mathrm{lcu}-1}}}(\gamma)\nonumber\\
&\mathrm{PREP}^{\dagger}\prod_{j=0}^{n_{\mathrm{lcu}}-1}
\left(|0\rangle\langle0|_j \otimes I+|1\rangle\langle1|_j \otimes V^{2^j}(\gamma)\right)\mathrm{PREP}\ e^{i A_2\gamma },
\label{eq:LCU_representation3}
\end{align}

It is straightforward to verify that the evolution operators appearing in Eq.~\ref{eq:LCU_representation2} take the form
\begin{widetext}
\begin{align}
V(\gamma)=e^{i \gamma\left( A_1 + {\rm I}/(4l)\right)}&=e^{i \frac{\gamma}{4l}{\rm I}} \exp\left(i\frac{\gamma}{4l}\sigma_{11}^{\otimes(n-1)}\otimes \left({\rm H} {\rm Z}{\rm H}\right) \right)
={\rm Ph}\left(\frac{\gamma}{4l}\right){\rm H}_1 {\rm MCRZ}_{1}^{2,\ldots, n}\left(-\frac{\gamma}{2l} \right) {\rm H}_1,\\
e^{i A_2 \gamma}&\approx e^{-i(iD^{\pm})\gamma}
\exp\left(i\frac{\gamma}{4l}\sigma_{11}^{\otimes(n-1)}\otimes \left({{\rm P}\left(-\frac{\pi}{2}\right) \rm H} {\rm Z}{\rm H}{\rm P}\left(-\frac{\pi}{2}\right)^\dagger \right) \right)\\
&=e^{D^{\pm}\gamma} {\rm P}_1\left(-\frac{\pi}{2}\right) {\rm H}_1\ {\rm MCRZ}_{1}^{2,\dots,n}\left(-\frac{\gamma}{2l} \right) {\rm H}_1 {\rm P}_1\left(-\frac{\pi}{2}\right)^\dagger,
\label{eq:LCU_representation2}
\end{align}
\end{widetext}
where ${\rm Ph}(\theta)$ is a global phase gate $e^{i\theta}$ and $e^{D^{\pm}\gamma}$ is given in Eq.~\ref{eq:bc1}. In general, a global phase gate does not affect the physical state, however, in the present setting, due to the controlled implementation of the unitary operators from the ancilla register within the LCU procedure, this global phase transforms into a phase gate and contributes non-trivially to the overall weighted superposition. It is also evident from Eq.~\ref{eq:LCU_representation2} that $V^{2^j}(\gamma)=V(2^j\gamma)$, which substantially reduces the implementation cost by enabling realization with a single ${\rm MCRZ}$ gate instead of a sequential chain of such gates. The quantum circuits of these operators are shown in Fig.~\ref{fig:diffoperatorsDNconds}. The complete quantum circuit corresponding to the differential operator in Eq.~\ref{eq:LCU_representation3} is presented in Fig.~\ref{fig:LCU_decompozition}.

Finally, the error arising from the Trotter approximation of operator $e^{D^{\pm}_{DN}\gamma}$ is bounded as
\begin{align}\label{eq:upper_error_diff}
     \lVert e^{D^{\pm}_{DN}\gamma} - {\rm U}(\gamma)\rVert\leq \frac{7}{32}\frac{\gamma^2 M}{l^2} + \frac{\gamma^2}{16 l^2} + \frac{\gamma^2(n-1)}{8l^2}.
 \end{align}
The first term defines the error caused by the Trotterization between the operators $A_1$ and $A_2$, taking into account the LCU procedure, as expressed in Eq.~\ref{eq:LCU_representation}, where the factor $M$ arises from summing the individual errors over all discretization points $k$. The second and third terms define the Trotter errors arising from the decomposition of the operator $e^{iA_2\gamma}$ in Eq.~\ref{eq:LCU_representation2}.

Note that placement of controls in the ${\rm MCRZ}$ gate is different from the ones in the $D^\pm$ operator previously. Here, they define the position of the Neumann boundary conditions. With controls on the all-$1$ state this imposes such boundary condition on the right end of the interval. Setting all controls to $0$ will define the Neumann boundary conditions on the left end of the interval. For an object inside the environment the bitstring used as the control will be defined by the placement of the object's edges. For multiple Neumann boundary conditions multiple circuits are added for each of its placements. Note also that $e^{i k_a\left( A_1 + {\rm I}/(4l)\right) \gamma}$ is implemented with linear complexity. This gate must also be controlled on the LCU ancillary register, therefore the total gate complexity of implementing a single Neumann boundary condition is $\mathcal{O}(nn_{\mathrm{lcu}})$.

The detailed recipe described in this appendix allows to straightforwardly implement these more general boundary conditions in the LEE system, for instance imposing zero Neumann boundary conditions on the density at the objects's edges.

Furthermore, the very same prescription can be applied to the non-conservative regime in LEE, with $c\neq 1/\bar{\rho}$. This is achieved following Eqs.~\ref{eq:hermeitianoperatorsDNcond2}--\ref{eq:LCU_representation3}, with the only difference being a different spectral shift to $A_1$, which must be used to make the LEE matrix positive semi-definite. In this case the shift is equal to $|1/\bar{\rho}-\bar{\rho}c^2|/(\sqrt{2}l){\rm I}$.

\end{document}